\documentclass[10pt,twocolumn,aps,english,pra,superscriptaddress,showpacs]{revtex4-1}
\usepackage{amsfonts}
\usepackage[T1]{fontenc}
\usepackage[latin9]{inputenc}
\usepackage{amsmath}
\usepackage{amssymb}
\usepackage{graphicx}
\usepackage{babel}
\usepackage{booktabs}
\usepackage{tabu}
\usepackage{subfigure}
\usepackage{graphicx}
\usepackage{mathrsfs}
\usepackage{color}

\makeatletter
\@ifundefined{textcolor}{}
{
 \definecolor{BLACK}{gray}{0}
 \definecolor{WHITE}{gray}{1}
 \definecolor{RED}{rgb}{1,0,0}
 \definecolor{GREEN}{rgb}{0,1,0}
 \definecolor{BLUE}{rgb}{0,0,1}
 \definecolor{CYAN}{cmyk}{1,0,0,0}
 \definecolor{MAGENTA}{cmyk}{0,1,0,0}
 \definecolor{YELLOW}{cmyk}{0,0,1,0}
 }
\makeatother

\newtheorem{lemma}{Lemma}
\newtheorem{proposition}{Proposition}

\begin{document}

\title{How to quantify direct correlations between variables}

\author{Shengjun Wu$^{1,*}$, and Jeffery Wu$^{1,2}$ \\
{\small\itshape $^1$National Laboratory of Solid State Microstructures and School of Physics,} \\
{\small\itshape and Collaborative Innovation Center of Advanced Microstructures, Nanjing University, Nanjing 210093, China}\\
{\small\itshape $^2$School of Physics and Astronomy, Shanghai Jiao Tong University, Shanghai 200240, China}\\[2pt]
{\small$^*$sjwu@nju.edu.cn} }

\date{\today}


\begin{abstract}
A crucial question throughout statistics is whether an observed correlation
between two variables is a direct correlation or only an indirect one mediated by a
confounder. We organize the existing nonlinear measures of direct correlation into two
families, each with a systematic construction: (i) removing the direct correlation from
the joint distribution and quantifying the resulting distributional shift, and (ii)
intervening on one variable via do-calculus and quantifying the response of the other.
For every Kullback-Leibler-based measure in either family we propose a Jensen-Shannon-based
regularized analogue; the regularized measures take values in $[0,1]$, satisfy the metric
property, and are free of the singularities of the Kullback-Leibler divergence.
We analyze the achievable upper bound of each regularized measure under the observed
marginals, and derive the maximal value each measure can attain when only the
alphabet sizes of the variables are fixed; the maxima admit closed forms built on a single
binary-entropy function. The measures are compared on a decision-making model and on
three public datasets (Titanic survival, UCI Adult income, and the 1973 Berkeley
graduate admissions), with bootstrap confidence intervals for every reported value.
\end{abstract}

\maketitle

\section{Introduction}

An important mission of science is to discover the connections between various variables, and often what we have access to is only the observed data. For example, we may have a list recording health condition, age, weight, height, diet and habits for a large population, or a database of course grades and personal information for university graduates, and we are interested in the relation between habits and health in the first case and in the connections between course grades in the second. A central task of statistics, and of science more generally, is to quantify the correlations, direct correlations, and causal relations between variables.

The word \textit{correlation} was originally used to measure the strength and direction of the linear relationship between two variables, while the word \textit{association} refers to the presence of a certain (linear or nonlinear) relationship between two variables, without quantifying the strength or direction of the relationship. However, with the development of information theory, different correlation measures (especially entropic correlation measures) have been used to measure the strength of nonlinear relationship as well. Therefore, in this paper, we shall not try to distinguish between \textit{correlation} and \textit{association}, while we shall use the single term \textit{correlation} to represent the linear as well as nonlinear relationship between variables.

Correlation is very different from direct correlation. In the 1973 UC~Berkeley graduate-admission data~\cite{bickel1975sex}, $44.5\%$ of the $2691$ male applicants but only $30.4\%$ of the $1835$ female applicants were admitted, an apparent $14$-point gender gap. Broken down by the six largest departments, however, the admission rates are close for the two sexes and in four of the six departments in fact slightly favour women; the overall gap arises because women applied disproportionately to the more competitive departments, not from a direct sex effect on admission. A non-trivial direct correlation can also coexist with a covariate whose own association with $Y$ differs substantially from the direct $X$--$Y$ effect. In the 1912 Titanic passenger record~\cite{titanic_kaggle}, for instance, the overall marginal correlation between passenger class and survival is strong, and stratifying by sex --- itself a strong predictor of survival --- leaves the class--survival association within each stratum essentially unchanged (Simpson's 1951 analysis of 2$\times$2$\times$2 contingency tables~\cite{Simpson1951} formalises such stratification). Distinguishing total from direct correlation is therefore essential before drawing scientific conclusions from observed correlations; we return to both datasets in Sec.~\ref{sec:applications}.

Quantifying correlation has a long history. Pearson's correlation coefficient captures linear dependence and Spearman's rank correlation~\cite{sedgwick2014spearman} extends it to monotonic dependence. Shannon's mutual information~\cite{shannon1948,shannon19482}, built on the Kullback-Leibler (KL) divergence~\cite{kullback1951}, captures arbitrary linear and nonlinear dependence. Normalised nonlinear coefficients include distance correlation~\cite{Szekely2007measuring,Kosorok2009brownian}, the maximal information coefficient~\cite{reshef2011detecting,Kinney2014equitability}, and the Hilbert--Schmidt independence criterion~\cite{Gretton2005HSIC}; Ref.~\cite{Altman2015points} gives a pedagogical overview of the association/correlation/causation hierarchy.

Disentangling direct from indirect correlation requires more than the bivariate distribution. Reichenbach's common-cause principle~\cite{reichenbach1956direction} states that a marginal correlation between $X$ and $Y$ may arise from $X\to Y$, $Y\to X$, or a common cause $Z$; conditional-independence criteria on directed acyclic graphs~\cite{hausman1999independence,GEIGER19903,lauritzen1990independence,VermaPearl1990,Koller2009probabilistic} formalise this distinction. For linear data the partial correlation coefficient~\cite{de2004discovery} is a mainstay of gene-network reconstruction~\cite{stuart2003gene}. For nonlinear data the conditional mutual information (CMI) plays the same role in systems biology~\cite{Zhang2012inferring,Liang2008gene,Zhang2015conditional}. Variants introduced to correct artefacts of CMI include the part mutual information (PMI)~\cite{Zhao2016part}, the multiscale association analysis~\cite{Shi2018quantifying}, and the independent conditional mutual information (ICMI)~\cite{FRzhao23}, which removes indirect paths in a two-step procedure. An alternative, intervention-based route is provided by Pearl's do-calculus on Bayesian networks~\cite{Pearl1988probabilistic,Pearl2009causality,Pearl2009causal,Pearl2022direct}, which yields the average causal effect (ACE) of Holland~\cite{Holland1986causal} and axiomatic quantifications of causal influence~\cite{Janzing2013quantifying}; textbook treatments are given in Refs.~\cite{Freedman2009statistical,McNamee2003confounding,Pearl2022direct,Peters2017Elements}.

Two issues cut across the existing nonlinear measures. First, those built on the Kullback-Leibler divergence are unbounded, and are singular whenever the reconstructed distribution has a zero where the original does not --- a situation common with sparse data. Second, the proposals in the literature come in seemingly unrelated flavours. We address both issues here. We show that the existing nonlinear measures fall into two families: a \emph{removal-of-direct-correlation} family (CMI, PMI, ICMI) and a \emph{do-calculus} family (ACE, Normalized ACE, do-based mutual information). For every Kullback-Leibler-based member of either family we introduce a Jensen-Shannon-based regularised analogue; since the square root of the Jensen-Shannon (JS) divergence $\sqrt{D_{JS}}$ is a metric in $[0,1]$~\cite{Lin1991,endres2003metric}, the regularised measures are bounded and singularity-free. We further analyze the \emph{achievable upper bound} of each regularised measure under the observed marginals, which is in general strictly below $1$ and depends on the alphabet size. The measures are compared on a decision-making toy model (Sec.~\ref{sec:dm}) and on three public real datasets (Sec.~\ref{sec:applications}), with bootstrap $95\%$ confidence intervals.

\section{Correlation}

We first consider linear correlations.
The standard quantity to measure the linear correlation between two variables $X$ and $Y$ is the Pearson's correlation coefficient (PCC)~\cite{Freedman2009statistical}
\begin{equation}
\mathcal{C}_{pcc}(X:Y)=\frac{\sum_{i}(x_i-\bar{x})(y_i-\bar{y})}
{\sqrt{\sum_i(x_i-\bar{x})^2\sum_i(y_i-\bar{y})^2}}
\label{eq1}
\end{equation}
where $x_i$ and $y_i$ are the $i$-th observed values of $X$ and $Y$ respectively. The PCC can be equivalently written as
\begin{equation}
\mathcal{C}_{pcc}(X:Y)
=\frac{\overline{xy} -(\overline{x})( \overline{y})}
{\sqrt{(\overline{x^2} -(\overline{x})^2 ) (\overline{y^2} -(\overline{y})^2 ) }}
\label{eq11}
\end{equation}
where $\overline{w}$ is the average value of $w$, namely, $\overline{x}=\sum_{x}xp(x)$, $\overline{y}=\sum_{y}yp(y)$, $\overline{x^2}=\sum_{x}x^2 p(x)$, $\overline{y^2}=\sum_{y}y^2 p(y)$, and $\overline{xy}=\sum_{xy}xy p(x,y)$.
Here $p(x,y)$ is the joint probability distribution of the two variables with $p(x)=\sum_y p(x,y)$ and $p(y)=\sum_x p(x,y)$ the marginal distributions respectively.
The value of PCC is between $-1$ and $1$, and its absolute value represents the correlation strength between two variables.

Correlation between two variables X and Y is not necessarily the direct correlation, as their correlation may originate from a common source, i.e., another variable Z.
When we deal with linear correlations, the partial correlation (PC) $\mathcal{C}^D_{pc}$ \cite{de2004discovery} can measure the direct linear correlation between X and Y,
{\small
\begin{equation}
\mathcal{C}^D_{pc}(X:Y|Z)=\frac{\mathcal{C}_{pcc}(X:Y)-\mathcal{C}_{pcc}(X:Z)\mathcal{C}_{pcc}(Y:Z)}{\sqrt{(1-(\mathcal{C}_{pcc}(X:Z))^2)(1-(\mathcal{C}_{pcc}(Y:Z))^2)}}
\label{pcdef}
\end{equation}  }
where $\mathcal{C}_{pcc}(X:Y)$ is the PCC between $X$ and $Y$, similarly for other notations.
Similar to the PCC, the PC $\mathcal{C}^D_{pc}$ takes values from $-1$ to 1.

Although PCC is a very good measure of linear correlation, it fails to measure nonlinear correlations; and similarly, the PC $\mathcal{C}^D_{pc}$ also fails to measure nonlinear direct correlation.
A nice framework to study nonlinear as well as linear correlations is to convert the database into a joint probability distribution $p(x,y)$ for two variables X and Y; or a multi-party probability distribution $p(x,y,z,...)$.
In the rest of the paper, our starting point is the joint probability distribution.

Given a joint probability distribution $p(x,y)$ of two variables $X$ and $Y$, we have a number of measures of the correlation between them, without any information about direct or causal relationships. From the joint probability distribution $p(x,y)$ one can easily obtain the marginal probability distributions $p(x)=\sum_y p(x,y)$ for $X$ and $p(y)=\sum_x p(x,y)$ for $Y$, as well as the conditional probabilities $p(y|x)=p(x,y)/p(x)$ and $p(x|y)=p(x,y)/p(y)$.
A good measure $\mathcal{C}(X:Y)$ of the correlation between $X$ and $Y$ for a given joint probability distribution $p(x,y)$ should satisfy the following natural properties: \\
\noindent (1) $\mathcal{C}(X:Y)=0$ if and only if $X$ and $Y$ are independent, i.e., $p(x,y)=p(x)p(y)$. \\
\noindent (2) $\mathcal{C}(X:Y)$ attains its maximum, over all joint distributions with a given pair of marginals $p(x),p(y)$, at a ``perfect correlation'' in which $X$ completely determines $Y$ (i.e., $p(y|x)=\delta_{y,f(x)}$ for some function $f$) or $Y$ completely determines $X$ (i.e., $p(x|y)=\delta_{x,g(y)}$ for some function $g$). Note that the value of this maximum depends in general on the marginals and on the alphabet sizes, and need not equal the upper bound $1$ of the measure's nominal range; see Sec.~\ref{sec:upperbound} for the achievable upper bounds of the regularized measures introduced in this paper. \\
\noindent (3) $\mathcal{C}(X:Y)$ is (or can be) normalized to take values in $[0,1]$ (or in $[-1,1]$ for a signed measure); the range $[-1,1]$ is used for linear measures such as the PCC. When the directional split matters we use $\mathcal{C}(X\rightarrow Y)$ and $\mathcal{C}(X\leftarrow Y)$. \\

A well-known quantity that gives a nice measure of linear as well as nonlinear correlation is the mutual information $H(X:Y)$, which is defined as
\begin{equation}
H(X:Y)=H(X)+H(Y) -H(X,Y)
\label{eqdefmutinf}
\end{equation}
where the marginal Shannon entropy $H(X)= \sum_x -p(x) \log_2 p(x)$ ($H(Y)= \sum_y -p(y) \log_2 p(y)$) depends only on the marginal probability distribution of a single variable $X$ ($Y$), and the total Shannon entropy $H(X,Y)= \sum_{xy} -p(x,y) \log_2 p(x,y)$ depends only on the joint probability distribution.
Throughout the paper all entropies and divergences are measured in bits (logarithms to base $2$), and we keep the notation $H(X{:}Y)$ for the mutual information, which is often written as $I(X;Y)$ in the statistical literature.
The mutual information can also be rewritten as
\begin{equation}
H(X:Y)=H(Y) -H(Y|X)
\label{eqdefmutinf1}
\end{equation}
where the conditional entropy $H(Y|X)=H(X,Y)-H(X)$ denotes the residual uncertainty of $Y$ given $X$, which is nonnegative and cannot exceed $H(Y)$, the overall uncertainty of $Y$. A natural idea to measure the amount of correlation that $X$ can influence $Y$ is defined as the ratio of the mutual information $H(X:Y)$ to the total uncertainty $H(Y)$ of $Y$,
\begin{equation}
\mathcal{C}_{mi}(X\rightarrow Y) = H(X:Y)/H(Y)
\label{Corr(XtoY)}
\end{equation}
which reaches the minimum value $0$ when $Y$ is independent of $X$, and the maximum value $1$ when $Y$ is completely determined by $X$ ($H(Y|X)=0$).
Similarly, the amount of correlation that $Y$ can influence $X$ is defined as the ratio of the mutual information $H(X:Y)$ to the total uncertainty $H(X)$ of $X$,
\begin{equation}
\mathcal{C}_{mi}(Y\rightarrow X)= \mathcal{C}_{mi}(X\leftarrow Y) = H(X:Y)/H(X)
\label{Corr(YtoX)}
\end{equation}
which reaches the minimum value $0$ when $X$ is independent of $Y$, and the maximum value $1$ when $X$ is completely determined by $Y$ ($H(X|Y)=0$).
We use the notation $\mathcal{C}_{mi}(X:Y)$ to denote the larger one,
\begin{equation}
\mathcal{C}_{mi}(X: Y) = \max \{ \mathcal{C}_{mi}(X\rightarrow Y), \mathcal{C}_{mi}(Y\rightarrow X) \}.
\label{Corr(X:Y)}
\end{equation}

We can remove the correlation between X and Y, therefore reconstruct a joint probability distribution $q(x,y)=p(x)p(y)$ of two independent variables.
The mutual information is actually the Kullback-Leibler (KL) divergence between the original probability distribution $p(x,y)$ and the reconstructed probability distribution $p(x)p(y)$, namely
\begin{eqnarray}
H(X:Y) &= & D_{KL} (p(x,y)||p(x)p(y)) \nonumber \\
& = & \sum_{xy} p(x,y) \log_2 \frac{p(x,y)}{p(x)p(y)} .
\end{eqnarray}
In order to have a measure between 0 and 1, we can alternatively define the regularized mutual information as the square root of the Jensen-Shannon (JS) divergence between $p(x,y)$ and $p(x)p(y)$,
\begin{equation}
\hat{\mathcal{C}}_{rmi} (X:Y)= \sqrt{ D_{JS} (p(x,y)||p(x)p(y)) } .  \label{defregmutinf}
\end{equation}
The JS divergence $D_{JS}(p(x)||q(x))$ of two probability distributions $p(x)$ and $q(x)$ is defined as the average of the KL divergences from each distribution to the mean distribution $m(x)=(p(x)+q(x))/2$,
\begin{equation}
D_{JS}(p(x)||q(x)) = (D_{KL}(p(x)||m(x))+D_{KL}(q(x)||m(x)))/2 ,
\end{equation}
which is symmetric in its two arguments and can be written in terms of entropies as
\begin{equation}
D_{JS}(p(x)||q(x)) = H(m(x)) - (H(p(x))+H(q(x)))/2 .
\end{equation}
The JS divergence is bounded by $\log 2$, i.e.\ by $1$ in the base-$2$ logarithm convention used throughout this paper, and this bound is attained only when the two distributions have disjoint supports. We use the square root $\sqrt{D_{JS}}$ because it is a proper metric on probability distributions~\cite{Lin1991,endres2003metric}, bounded in $[0,1]$, and free of the singularities of the KL divergence.

\section{Direct correlation}

When we have access to more than two variables, say three variables $X$, $Y$, and $Z$, we have more correlations to consider. Even focusing on the correlation between $X$ and $Y$, we may want to separate the amount of direct correlation from the amount of indirect correlation mediated by a third variable $Z$. Given a joint distribution $p(x,y,z)$ we can reliably quantify both the total and the direct correlation between $X$ and $Y$, while causal correlation is in general under-identified from observational data alone.

Neither correlation nor direct correlation implies causation. We use the term \emph{direct correlation} for the $X$--$Y$ dependence that survives after the indirect path through $Z$ is removed, and \emph{causal correlation} for the dependence induced by an intervention on one variable; the latter may also have direct and indirect components.

From the joint probability distribution $p(x,y,z)$, we can sum over $z$ to obtain the marginal probability distribution $p(x,y)$, as $p(x,y)=\sum_z p(x,y,z)$, and similarly for $p(x,z)$ and $p(y,z)$.  From $p(x,y)$, we have all kinds of measures of correlation between $X$ and $Y$ as discussed above. However, the correlation between $X$ and $Y$ may be due to or partially due to a common parent variable (or an intermediary variable) $Z$. Can we separate the influence directly between $X$ and $Y$ from that via a third variable $Z$?

In the rest of the paper, we shall focus on quantifying the direct correlation between $X$ and $Y$.
When we deal with nonlinear correlations, the direct correlation between two variables becomes more subtle. In order to define meaningful measures, we introduce two systematic ways to do this.

\subsection{First strategy via removal of direct correlation}
\label{sec:measures_removal}

The first method to construct a measure of direct correlation between X and Y, intuitively, is to find how much a probability distribution has to change if the direct correlation between X and Y is removed.
From the given joint probability distribution $p(x,y,z)$, we can construct
another joint probability distribution $q(x,y,z)$ that contains no direct correlation between X and Y, but otherwise is as close as possible to the original distribution $p(x,y,z)$. Then we define the distance between the original $p(x,y,z)$ and the new $q(x,y,z)$ as a measure of direct correlation between X and Y.

From the definition of conditional probability $p(x,y|z) = p(x,y,z)/p(z)$, we know that
\begin{equation}
p(x,y,z) = p(x,y|z) p(z)
\end{equation}
where the conditional probability $p(x,y|z)$, a joint probability distribution of X and Y for each fixed value of Z ($Z=z$), is in general not equal to $p(x|z) p(y|z)$.
The reconstructed conditional probability $q(x,y|z) =p(x|z) p(y|z)$ indicates that all direct correlations between X and Y are removed, therefore, we can construct a new joint probability distribution $q(x,y,z)$ as
\begin{equation}
q(x,y,z) = q(x,y|z) p(z) = p(x|z) p(y|z) p(z)   \label{defqxyz}
\end{equation}
where $p(x|z)=p(x,z)/p(z)$ and $p(y|z)=p(y,z)/p(z)$ are evaluated from the original probability distributions.
The reconstructed joint probability distribution $q(x,y,z)$ is constructed from the original one $p(x,y,z)$ with all direct correlation between X and Y removed (see Fig. \ref{fig_cmi}), therefore, a natural measure to quantify the direct correlation
between X and Y is the distance between these two joint probability distributions.
The conditional mutual information (CMI) is actually such a measure. It is the KL divergence of $q(x,y,z)$ from $p(x,y,z)$,
\begin{eqnarray}
\mathcal{C}^D_{cmi}(X:Y) &=& D_{KL}(p(x,y,z)||q(x,y,z))  \nonumber \\
                           &=& \sum_{x,y,z} p(x,y,z) \log_2 \frac{p(x,y,z)}{ q(x,y,z)}
\end{eqnarray}
which quantifies the information loss when $q(x,y,z)$ is used to replace $p(x,y,z)$.
The CMI is a good measure of direct correlation between X and Y, and one can easily show that
\begin{eqnarray}
&\mathcal{C}^D_{cmi}(X:Y) = \sum_{x,y,z} p(x,y,z) \log_2 \frac{p(x,y,z)}{ p(x|z) p(y|z) p(z)}  \nonumber \\
&=  H(X,Z)+H(Y,Z)-H(X,Y,Z)-H(Z)
\end{eqnarray}
which is the familiar form.
\begin{figure}[htbp]
\centering
\includegraphics[width=0.4\textwidth]{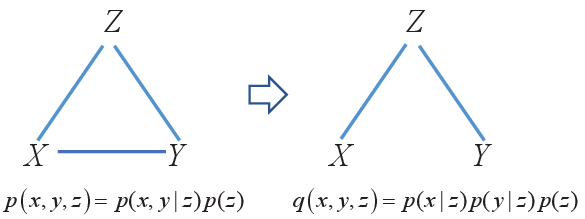}
\caption{The CMI measures the change of the joint distribution with the direct correlation between X and Y removed.}
\label{fig_cmi}\hspace{-2mm}
\end{figure}

However the measure of CMI is not normalized. For the convenience of comparison, we propose to use the JS divergence instead, and define the
normalized CMI $\mathcal{C}^D_{cmi,JS}(X:Y)$ as the JS divergence between $p(x,y,z)$ and $q(x,y,z)$.
Therefore, we propose the normalized CMI as
\begin{eqnarray}
&\mathcal{C}^D_{cmi,JS}(X:Y) = D_{JS}(p(x,y,z)||q(x,y,z)) \nonumber \\
&= H(m(x,y,z)) - (H(p(x,y,z))+H(q(x,y,z)))/2
\end{eqnarray}
with $m(x,y,z)=(p(x,y,z)+q(x,y,z))/2$ and $q(x,y,z)$ defined in (\ref{defqxyz}).
Since the square root of JS divergence is a good metric that satisfies the triangle inequality,
we also propose the regularized CMI as
\begin{equation}
\hat{\mathcal{C}}^D_{rcmi}(X:Y) = \sqrt{ D_{JS}(p(x,y,z)||q(x,y,z))}
\end{equation}
which is also a good measure of direct correlation between X and Y.

In \cite{Zhao2016part} a part mutual information (PMI) was introduced in a similar way. In order to reconstruct a joint probability $q(x,y,z)$ with no direct correlation between X and Y, one can replace $p(x|z)$ ($=\sum_y p(x|y,z)p(y|z)$) by $q(x|z) =\sum_y p(x|y,z)p(y)$, and replace $p(y|z)$ ($=\sum_x p(y|x,z)p(x|z)$)
by $q(y|z) =\sum_x p(y|x,z)p(x)$, i.e., a new joint probability distribution $q'(x,y,z)$ is constructed as (see Fig. \ref{fig_pmi})
\begin{eqnarray}
q'(x,y,z) &=& q(x|z) q(y|z) p(z) \nonumber \\
&=& p(z) \sum_y p(x|y,z)p(y)  \sum_x p(y|x,z)p(x)   \label{reconstructedjointprobforPMI}
\end{eqnarray}
The part  mutual information (PMI) is actually the KL divergence of $q'(x,y,z)$ from $p(x,y,z)$,
\begin{equation}
\mathcal{C}^D_{pmi}(X:Y) = D_{KL}(p(x,y,z)||q'(x,y,z)) .
\end{equation}
Similarly, we can define the normalized and regularized versions of PMI.

\begin{figure}[htbp]
\centering
\includegraphics[width=0.4\textwidth]{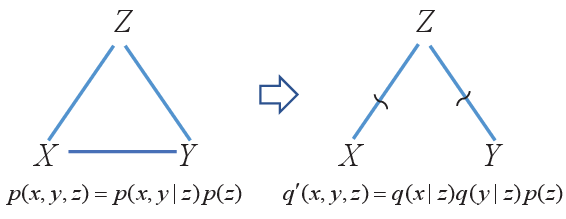}
\caption{The PMI measures the change of the joint distribution with the direct correlation between X and Y removed and the joint distribution carefully reconstructed as in (\ref{reconstructedjointprobforPMI}).}
\label{fig_pmi}\hspace{-2mm}
\end{figure}

A recent measure called the independent conditional mutual information (ICMI) \cite{FRzhao23} is also  introduced along this line, though in a two-step procedure.
From the joint probability distribution $p(x,y,z)$, one first removes the correlations between X and Z to form
the probability distribution $p_1(x,y,z)=p(y|x,z)p(x)p(z)$ (see Fig. \ref{fig_icmi}), which keeps the direct correlation between X and Y.  Furthermore, one removes the direct correlation between X and Y to form
another new probability distribution $p_2(x,y,z)=p(x)p(y,z)$. The one-way ICMI ($X\rightarrow Y$) is defined as the change of the probability distribution during the second step in which the direct correlation between X and Y is removed, i.e., the KL divergence of $p_2(x,y,z)$ from $p_1(x,y,z)$,
\begin{eqnarray}
\mathcal{C}^D_{icmi}(X\rightarrow Y) &=& D_{KL}(p_1(x,y,z)||p_2(x,y,z))  \nonumber \\
&=& D_{KL}(p(y|x,z)p(x)p(z)||p(x)p(y,z))
\end{eqnarray}
which denotes the one-way direct correlation from X to Y.
\begin{figure}[htbp]
\centering
\includegraphics[width=0.5\textwidth]{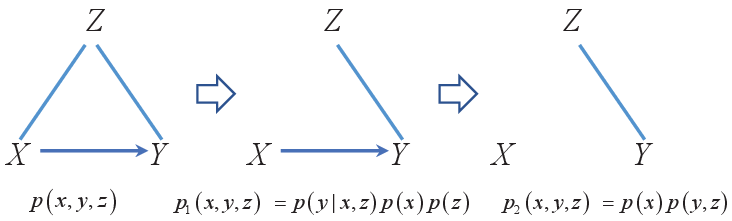}
\caption{The one-way ICMI $\mathcal{C}^D_{icmi}(X\rightarrow Y)$ measures the distance between a joint distribution $p_1$ with correlation between Z and X removed and another joint distribution $p_2$ with the correlation between X and Y further removed as well.}
\label{fig_icmi}\hspace{-2mm}
\end{figure}

Similarly, the one-way ICMI from Y to X is defined as
\begin{equation}
\mathcal{C}^D_{icmi}(X\leftarrow Y) = D_{KL}(p(x|y,z)p(y)p(z)||p(y)p(x,z)) .
\end{equation}
The average of the two one-way correlations is a measure of the direct correlation between X and Y,
\begin{equation}
\mathcal{C}^D_{icmi}(X : Y) = \left(\mathcal{C}^D_{icmi}(X\rightarrow Y) + \mathcal{C}^D_{icmi}(X\leftarrow Y) \right) /2 .
\end{equation}

The definitions of PMI and ICMI both rely on the KL divergence, which may have singularity problems for sparse data (some cases do not occur) and may not be bounded. These problems can be overcome if the KL divergence is replaced by the JS divergence, which gives a normalized distance between two probability distributions and has no singularity problems.
For all the direct correlation measures that rely on a distance measure between two probability distributions, we can always choose to replace it with the JS divergence to have regularized measures.

In particular, the regularized one-way ICMI from X to Y is defined as
\begin{eqnarray}
\hat{\mathcal{C}}^D_{ricmi}(X\rightarrow Y) = \sqrt{D_{JS}(p_1(x,y,z)||p_2(x,y,z)) } \nonumber \\
= \sqrt{D_{JS}(p(y|x,z)p(x)p(z)||p(x)p(y,z)) }  .
\end{eqnarray}
The regularized one-way ICMI from Y to X is defined as
\begin{eqnarray}
\hat{\mathcal{C}}^D_{ricmi}(Y\rightarrow X) = \sqrt{D_{JS}(p(x|y,z)p(y)p(z)||p(y)p(x,z)) } .
\end{eqnarray}
The regularized version for the two-way ICMI can be defined similarly,
\begin{equation}
\hat{\mathcal{C}}^D_{ricmi}(X : Y) = (\hat{\mathcal{C}}^D_{ricmi}(X\rightarrow Y) + \hat{\mathcal{C}}^D_{ricmi}(X\leftarrow Y) ) /2 . \label{defricmi2way}
\end{equation}

\subsection{Second strategy via do-calculus}
\label{sec:measures_do}

Another way to quantify the direct correlation between $X$ and $Y$ is to see how much the distribution of $Y$ changes in response to a freely chosen value $x$ of $X$. Pearl's do-calculus on directed acyclic graphs~\cite{Pearl2009causality,Pearl1988probabilistic} formalises this through the intervention operator $\mathrm{do}(X=x)$, which severs all arrows into $X$ and fixes $X$ to $x$. Throughout this subsection we assume that $Z$ is a sufficient back-door adjustment set for $X\to Y$~\cite{Pearl2009causality,Peters2017Elements}, i.e.\ $Z$ blocks every back-door path from $X$ to $Y$ and contains no descendant of $X$; under this condition
\begin{equation}
p(y| do (X=x)) = \sum_z p(y|x,z)p(z) ,
\end{equation}
which denotes a new reconstructed conditional probability distribution of $Y$ with fixed $x$ when all influence to X is removed (see Fig. \ref{fig_do}).
\begin{figure}[htbp]
\centering
\includegraphics[width=0.45\textwidth]{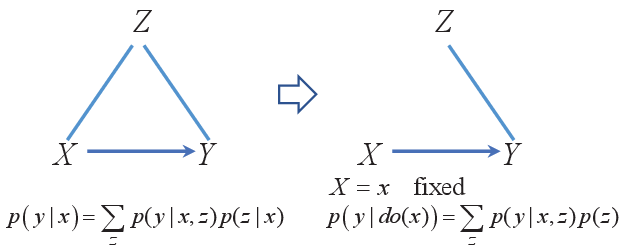}
\caption{The reconstructed conditional do-probability $p(y|do(x))$ has removed all influences to $X$, therefore, the change of the distribution $p(y|do(x))$ of Y in response to a change in $x$ represents the direct influence of X to Y.  The maximum distance between two distributions $p(y|do(x))$ and $p(y|do(x^{\prime}))$ of Y over any pair of values ($x$, $x^{\prime}$) of X is a natural way to quantify the direct influence of X on Y.}
\label{fig_do}\hspace{-2mm}
\end{figure}
The reconstructed conditional probability via do-calculus,
$p(y|\mathrm{do}(X=x)) \equiv p(y|\mathrm{do}(x))$, is different from the original conditional probability $p(y|x)$ and can be obtained from the latter by replacing $p(z|x)$ by $p(z)$ in the equality
$p(y|x)= \sum_z p(y|x,z)p(z|x)$; this replacement removes any dependence between $X$ and $Z$. How much the distribution $p(y|\mathrm{do}(x))$ of $Y$ changes in response to a change in $x$ represents the amount of direct correlation from $X$ to $Y$.
Namely, we can use the maximum distance between two distributions $p(y|do(x))$ and $p(y|do(x^{\prime}))$ of Y over any pair of values ($x$, $x^{\prime}$) of X as a measure of the direct correlation from X to Y,
\begin{equation}
\mathcal{C}^D_{do}(X \rightarrow Y) = \max_{x,x'} D_{ist} \left( p(y|do(x))  || p(y|do(x')) \right)
\end{equation}
where each different distance measure $D_{ist}$ gives a different quantitative measure of direct correlation between X and Y.

The average causal effect (ACE) of Holland~\cite{Holland1986causal} and Pearl~\cite{Pearl2009causality} is such a measure,
\begin{equation}
\mathcal{C}^D_{ace}(X \rightarrow Y) = \max_{x,x'} \max_y \left( p(y|\mathrm{do}(x)) -p(y|\mathrm{do}(x')) \right),
\end{equation}
which depends on a particular pair of values of $x$ and a particular value of $y$. We advocate instead the normalized average causal effect (NACE),
\begin{equation}
\mathcal{C}^D_{nace}(X \rightarrow Y)=
\max_{x,x'} \frac{1}{2}  \sum_y | p(y|\mathrm{do}(x)) -p(y|\mathrm{do}(x')) |  ,
\end{equation}
which integrates the total variation distance between the two intervened distributions of $Y$. Both ACE and NACE lie in $[0,1]$. When $Y$ is independent of $X$ under the intervention both measures give $0$, and when there exist $x\ne x'$ such that $p(y|\mathrm{do}(x))$ and $p(y|\mathrm{do}(x'))$ have disjoint supports, then $\mathcal{C}^D_{nace}=1$ while $\mathcal{C}^D_{ace}\leq 1$.
Furthermore, when we consider continuous variables,  NACE has a natural generalization while ACE does not,
\begin{equation}
\mathcal{C}^D_{nace}(X \rightarrow Y)=
 \max_{x,x'} \frac{1}{2}  \int_y \left| p(y|do(x)) -p(y|do(x')) \right| dy  .
\end{equation}
ACE depends on particular values of X and Y, while NACE takes account of all values of Y and it still depends on particular values of X.
The absolute value in the definition makes ACE and NACE not so smooth with respect to parameter changes. One can solve this problem by using the KL divergence to define the distance $D_{ist}$,
\begin{equation}
\mathcal{C}^D_{aceKL}(X \rightarrow Y) =  \max_{x,x'} D_{KL} \left( p(y|do(x)) || p(y|do(x')) \right) ,
\end{equation}
which is always nonnegative, but has no upper bound. In order to have a similar but normalized measure, we advocate using the JS divergence to define the distance $D_{ist}$.
As the square root of JS divergence is actually a good metric that satisfies the triangle inequality, we can actually define a regularized measure of direct correlation from X to Y as
\begin{equation}
\hat{\mathcal{C}}^D_{race}(X \rightarrow Y) =  \max_{x,x'} \sqrt{ D_{JS} \left( p(y|do(x)) || p(y|do(x')) \right) }
\end{equation}
which we shall refer to as the regularized average causal effect (RACE).

In the rest of the subsection, we propose an alternative way to define direct correlation via do-calculus.
From the reconstructed conditional probability $p(y|do(x))$, we can reconstruct a joint probability distribution $p_{do}(x,y)$ of two variables X and Y via
\begin{equation}
p_{do}(x,y)= p(y|do(x)) p(x)
\end{equation}
where $p(x)$ is given by the original marginal probability distribution. The reconstructed joint probability distribution $p_{do}(x,y)$ has all influences between X and Z removed, therefore any measure of correlation between X and Y for $p_{do}(x,y)$ is a measure of the direct correlation from X to Y in the original distribution. We propose another measure $\mathcal{C}^D_{mi,do}(X \rightarrow Y)$ of direct correlation via the mutual information between X and Y for $p_{do}(x,y)$ as
\begin{eqnarray}
& &\mathcal{C}^D_{mi,do}(X \rightarrow Y) =  \mathcal{C}_{mi}(X \rightarrow Y)|_{p_{do}(x,y)} \nonumber \\
&=& \left( H(p_{do}(x)) + H(p_{do}(y)) - H(p_{do}(x,y)) \right)/H(p_{do}(y))
\end{eqnarray}
where $p_{do}(x)=\sum_y p_{do}(x,y)$ and $p_{do}(y)=\sum_x p_{do}(x,y)$ are the marginal probability distributions of the reconstructed joint distribution $p_{do}(x,y)$.
Another regularized measure of the direct correlation between X and Y, i.e., the regularized mutual information from do-calculus, is defined as the regularized mutual information in the reconstructed joint probability distribution  $p_{do}(x,y)$,
\begin{equation}
\hat{\mathcal{C}}^D_{rmi,do} (X:Y)= \sqrt{ D_{JS} (p_{do}(x,y)||p_{do}(x)p_{do}(y)) } .
\end{equation}
where $p_{do}(x)=\sum_y p_{do}(x,y) = \sum_y p(y|do(x)) p(x) = \sum_y \sum_z p(y|x,z)p(z) p(x) = \sum_z p(z) p(x) = p(x)$, however, $p_{do}(y)=\sum_x p_{do}(x,y) = \sum_x p(y|do(x)) p(x) = \sum_x \sum_z p(y|x,z)p(z) p(x)$ which is different from $p(y)$ in general.

With more different measures of the distance between two probability distributions, we can have more different measures of direct correlation.  The measures mentioned above are quite natural, each with a direct intuitive physical meaning, there are many other choices that will not be discussed here.

\section{More discussions}

\subsection{Summary of the new measures}

In summary, we are given a certain database of three variables X, Y, and Z with a three-variable joint probability distribution $p(x,y,z)$, which is the only thing we have to start with.
From the three-variable joint probability distribution $p(x,y,z)$ the two-variable joint probability distribution $p(x,y)$ as well as other marginal probability distributions can easily be obtained.
Based on these probability distributions, we have proposed several normalized or regularized measures of correlation to describe total correlation as well as direct correlation between two variables.

In order to have a regularized measure of total correlation in the range $[0,1]$,
we have proposed the regularized mutual information
\begin{equation}
\hat{\mathcal{C}}_{rmi} (X:Y)= \sqrt{ D_{JS} (p(x,y)||p(x)p(y)) }
\end{equation}
to measure the total correlation between X and Y.

In order to quantify the direct correlation between X and Y, we have also proposed several measures of direct correlation, either normalized or regularized.
These measures include the regularized conditional mutual information
\begin{eqnarray}
&\hat{\mathcal{C}}^D_{rcmi}(X:Y) = \sqrt{ D_{JS}(p(x,y,z)||q(x,y,z))} \nonumber  \\
                    &= \sqrt{ D_{JS}(p(x,y,z)|| p(x|z) p(y|z) p(z) )}
\end{eqnarray}
the regularized one-way ICMI from X to Y
\begin{eqnarray}
\hat{\mathcal{C}}^D_{ricmi}(X\rightarrow Y) = \sqrt{D_{JS}(p_1(x,y,z)||p_2(x,y,z)) } \nonumber \\
= \sqrt{D_{JS}(p(y|x,z)p(x)p(z)||p(x)p(y,z)) }
\end{eqnarray}
the regularized two-way ICMI
\begin{equation}
\hat{\mathcal{C}}^D_{ricmi}(X : Y) = (\hat{\mathcal{C}}^D_{ricmi}(X\rightarrow Y) + \hat{\mathcal{C}}^D_{ricmi}(X\leftarrow Y) ) /2
\end{equation}
the normalized average causal effect
\begin{equation}
\mathcal{C}^D_{nace}(X \rightarrow Y) = \max_{x,x'} \frac{1}{2}  \sum_y | p(y|do(x)) -p(y|do(x')) |
\end{equation}
the regularized average causal effect
\begin{equation}
\hat{\mathcal{C}}^D_{race}(X \rightarrow Y) =  \max_{x,x'} \sqrt{ D_{JS} \left( p(y|do(x)) || p(y|do(x')) \right) }
\end{equation}
and the regularized mutual information from do-calculus
\begin{equation}
\hat{\mathcal{C}}^D_{rmi,do} (X:Y)= \sqrt{ D_{JS} (p_{do}(x,y)||p(x)p_{do}(y)) } .
\end{equation}
We need not return to the previous section for the definitions if we note that $p(y|do(x)) =\sum_z p(y|x,z) p(z)$, $p_{do}(x,y) =p(y| do(x)) p(x)$ and $p_{do}(y) = \sum_x p_{do}(x,y)$.

Each of the above new measures has a value between 0 and 1, with the value 0 indicating a zero (direct) correlation, and the value 1 is an upper bound for maximum (direct) correlation (see Sections~\ref{sec:upperbound} and \ref{sec:alphabetmax} for more discussions).
Furthermore, each one can be easily extended to the case of continuous variables with the summations replaced by integrals, and each measure still lies in between the range $[0,1]$ for the continuous variable case.

\subsection{A faithful measure of direct correlation requires a prior influence framework}

The measures above have been introduced without any prior influence framework, i.e.\ a graph specifying the possible causal relationships among the variables. Here we show, via a standard example, that any measure built from the joint distribution alone cannot distinguish certain inequivalent causal structures. Consider the two influence diagrams in Fig.~\ref{fig_twononseparablemodels}.
\begin{figure}[htbp]
\centering
\includegraphics[width=0.35\textwidth]{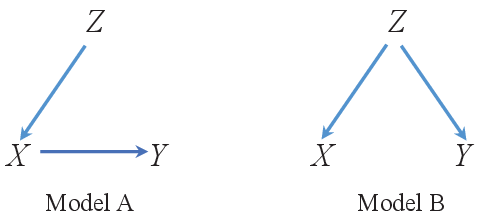}
\caption{Two Markov-equivalent~\cite{VermaPearl1990} influence diagrams yielding the same joint distribution $p(x,y,z)=\frac{1}{2}\delta_{y,x}\delta_{x,z}$.}
\label{fig_twononseparablemodels}\hspace{-2mm}
\end{figure}
In model~A, $p(z)=1/2$ on $z\in\{0,1\}$, $X$ is determined by $Z$ via $p(x|z)=\delta_{x,z}$, and $Y$ is determined by $X$ via $p(y|x)=\delta_{y,x}$; the joint distribution factorizes as $p(x,y,z)=p(y|x)p(x|z)p(z)$. In model~B, the same $p(z)$ is used, and both $X$ and $Y$ are determined by $Z$. The two models give the identical joint distribution, but $X$ has a direct influence on $Y$ in model~A while it has none in model~B. Any measure that depends only on $p(x,y,z)$ is therefore blind to this distinction. This is the observational-equivalence phenomenon of Verma and Pearl~\cite{VermaPearl1990}: two directed acyclic graphs with the same Markov equivalence class cannot be distinguished by observational data alone. A faithful quantification of direct correlation therefore requires a prior influence diagram, obtained from domain knowledge or from experiments.

\subsection{Strategies for sparse data}
\label{sec:sparse}

In a real situation, from the observed data we have a joint probability distribution $p(x,y,z)$, which is the starting point for our investigation. The above measures of direct correlation depend on comparing a reconstructed probability distribution and an original one, in defining a reconstructed distribution we usually need to use some conditional probability distribution calculated from the original joint distribution. However, the observed data may be sparse data that sometimes gives no incidences for some combinations of the variables. For example, in the cases shown in Fig. \ref{fig_twononseparablemodels}, the combinations of $(0,1)$ and $(1,0)$ for $(x,z)$ never occur.
In the definition of $p(y|do(x))=\sum_z p(y|x,z)p(z)$ we need to calculate the conditional probability $p(y|x,z)$ from the original data by $p(y|x,z)= p(x,y,z)/p(x,z)$, which is clearly ill-defined as $p(x,z)$ is $0$ for the combinations of $(0,1)$ and $(1,0)$.

To make sure the reconstructed conditional distributions in do-calculus are well defined, we propose three strategies. At every $(x,z)$ with $p(x,z)=0$ we replace the undefined $p(y|x,z)$ by (a) the uniform distribution $1/d_Y$ ($d_Y$ is the number of possible values of Y), (b) the marginal $p(y)$, or (c) the conditional marginal $p(y|z)$. Each of these choices defines a valid conditional distribution of $Y$ for every fixed $x$; the three are ordered from least informative (a) to most informative (c), and reflect different implicit priors on the missing cells. Strategy (a) only works for discrete random variables, while strategies (b) and (c) work for both discrete and continuous random variables.

For the two models in Fig.~\ref{fig_twononseparablemodels}, strategies (a) and (b) both give $p(0|do(0)) =p(1|do(1))=\frac{3}{4}$, $p(0|do(1)) =p(1|do(0)) =\frac{1}{4}$, and
\begin{eqnarray}
\mathcal{C}^D_{ace}(X \rightarrow Y) &=& \mathcal{C}^D_{nace}(X \rightarrow Y) = \tfrac{1}{2},\\
\mathcal{C}^D_{mi,do}(X \rightarrow Y) &=& \tfrac{3}{4} \log_2 3 - 1 \approx 0.189 ;
\end{eqnarray}
whereas strategy (c) gives $p(y|do(x)) =\frac{1}{2}$ and
\begin{equation}
\mathcal{C}^D_{ace} = \mathcal{C}^D_{nace} = \mathcal{C}^D_{mi,do} = 0 .
\end{equation}
Strategy (c) therefore matches CMI (which is zero in this case) and is the most conservative in that it does not manufacture any direct correlation from the missing cells. In this paper, unless stated otherwise, we use strategy (b), which adds the minimum amount of prior information consistent with the marginal of $Y$.

The definitions of PMI, ICMI and their regularized versions also rely on reconstructed probability distributions that may have similar singularity problems, we can use the same three strategies to calculate PMI, ICMI and the corresponding regularized measures for sparse data. However, CMI is always well defined, it never has such singularity problems even for sparse data.
The direct correlation between X and Y (in either model A or model B) in Fig. \ref{fig_twononseparablemodels} is $0$ in terms of CMI; and is also $0$ in terms of PMI and ICMI and measures from do-calculus if we take singularity strategy (c), but nonzero if we take singularity strategy (a) or (b) for sparse data. From this we know that strategy (c) for sparse data is very strict and tends not to add any possible direct correlation in reconstructing a complete probability distribution from an incomplete one.

It is not possible to consider all the measures, instead in the rest of the paper, we shall focus on the normalized or regularized measures of direct correlation.

\subsection{Upper bounds of the regularized measures}
\label{sec:upperbound}

Every regularized measure introduced in Sec.~\ref{sec:measures_removal}--\ref{sec:measures_do} is built on the square root of a JS divergence between two probability distributions, and therefore lies in $[0,1]$~\cite{Lin1991,endres2003metric}. The theoretical value $1$ is attained only when the two distributions being compared have disjoint supports. This is an extreme condition that is almost never reached in practice, for two reasons: (i)~for each measure the two distributions involved share, by construction, a large portion of their support, and (ii)~the attainable maximum depends on the marginals $p(x), p(y), p(z)$ and on the alphabet sizes $d_X=|X|, d_Y=|Y|, d_Z=|Z|$. We therefore distinguish the \emph{trivial upper bound} $1$ from the \emph{achievable upper bound}, i.e.\ the maximum of the measure over all joint distributions that share the observed marginals, and we give elementary analytical and numerical bounds below.
The complementary question --- the maximum of each measure when \emph{only the alphabet sizes} are fixed and the marginals are free --- is answered in Section~\ref{sec:alphabetmax}.

\paragraph{Regularized total mutual information $\hat{\mathcal{C}}_{rmi}$.} For the measure $\hat{\mathcal{C}}_{rmi}(X{:}Y)=\sqrt{D_{JS}(p(x,y)\|p(x)p(y))}$ one has the obvious inclusion
$\mathrm{supp}\,p(x,y)\subseteq\mathrm{supp}\,p(x)\times\mathrm{supp}\,p(y)=\mathrm{supp}\,p(x)p(y)$,
so the supports of the two distributions overlap on $\mathrm{supp}\,p(x,y)$ and the value $1$ cannot be attained. The explicit maximum under uniform marginals $p(x)=p(y)=1/d$ and deterministic coupling $Y=X$ can be computed in closed form,
\begin{equation}
\hat{\mathcal{C}}_{rmi}^{\max}(d)= \sqrt{ 1-\frac{1}{2}\log_2 \big(\frac{d+1}{d}\big) - \frac{1}{2d}\log_2 \big(d+1 \big) },  \label{maxrmi}
\end{equation}
which gives $0.558$ for $d=2$, $0.741$ for $d=4$, $0.910$ for $d=16$, and tends to $1$ as $d\to\infty$. For binary variables the regularized MI can therefore never exceed $\approx 0.56$, a fact that must be kept in mind when interpreting numerical values. For the case of continuous variables we have $\hat{\mathcal{C}}_{rmi}^{\max}(d)=1$.

\paragraph{Regularized conditional mutual information $\hat{\mathcal{C}}^D_{rcmi}$.} For $\hat{\mathcal{C}}^D_{rcmi}=\sqrt{D_{JS}(p(x,y,z)\|p(x|z)p(y|z)p(z))}$ the two distributions share the marginals $p(x,z)$ and $p(y,z)$, hence also the support of $p(x,z)\cap p(y,z)$. Inside every stratum $Z=z$ the argument of the previous paragraph applies with $p(x|z)$ and $p(y|z)$ in place of $p(x)$ and $p(y)$, so the stratum-wise maximum is controlled by $d_X$, $d_Y$ and the stratum-wise marginals. The overall upper bound can be computed by enumerating deterministic couplings $Y=f_z(X)$ within each $z$; for the three real datasets of Sec.~\ref{sec:applications} we report the resulting achievable upper bound alongside every measurement in Table~\ref{tab:realmeasures}.
If instead the marginals are left completely free and only the alphabet sizes $d_X= d_Y =d$ are fixed, the maximum of $\hat{\mathcal{C}}^D_{rcmi}$ equals the bound (\ref{maxrmi}); this and the analogous maxima for the other regularized measures are derived in Section~\ref{sec:alphabetmax}.

\paragraph{Regularized PMI and ICMI.} The PMI reconstructs $p(x|z)$ and $p(y|z)$ by marginalising $X$ and $Y$ under the \emph{unconditional} marginals of the partner variable (see Eq.~(\ref{reconstructedjointprobforPMI})), while the ICMI proceeds in two steps, first severing $X{-}Z$ and then $X{-}Y$. Neither reconstruction admits a simple closed-form upper bound, but in both cases $\sqrt{D_{JS}}\le 1$, and the achievable bound is again obtained by optimising over the deterministic couplings that share the observed $p(x,z)$; we report these bounds numerically in Sec.~\ref{sec:applications}. In general, the rPMI reconstruction shrinks the reconstructed joint towards the product of the unconditional marginals and therefore tends to yield larger upper bounds than the rCMI reconstruction, while the rICMI reconstruction involves different marginals on the two sides of the KL/JS symbol and the bound is asymmetric between the $X\to Y$ and $X\leftarrow Y$ directions.

\paragraph{NACE and RACE.} Both $\mathcal{C}^D_{nace}$ and $\hat{\mathcal{C}}^D_{race}$ compare distributions of $Y$ of dimension $d_Y$, not of the full $(X,Y,Z)$ triple. The NACE is a total variation distance between the two extreme $p(y|\mathrm{do}(x))$, $p(y|\mathrm{do}(x'))$; its maximum is $1$ and is attained whenever the two distributions have disjoint supports, which is possible already for $d_Y=2$. Similarly, $\hat{\mathcal{C}}^D_{race}\le 1$ with equality attainable under the same condition. In particular, NACE and RACE saturate at $1$ in the two-point limit where one value of $X$ deterministically drives $Y$ to one value and another value of $X$ deterministically drives $Y$ to the other value; this is what makes these two measures comparatively larger on datasets in which one specific contrast dominates (e.g.\ a two-arm clinical trial).

\paragraph{Do-based regularized mutual information $\hat{\mathcal{C}}^D_{rmi,\mathrm{do}}$.} The measure $\hat{\mathcal{C}}^D_{rmi,\mathrm{do}}=\sqrt{D_{JS}(p_{\mathrm{do}}(x,y)\|p_{\mathrm{do}}(x)p_{\mathrm{do}}(y))}$ has an alphabet-dependent upper bound identical in form to the one for $\hat{\mathcal{C}}_{rmi}$ derived above, but applied to the intervened joint distribution $p_{\mathrm{do}}$. Because $p_{\mathrm{do}}$ has a strictly smaller support of possible couplings than the empirical $p$, the achievable upper bound for $\hat{\mathcal{C}}^D_{rmi,\mathrm{do}}$ is generally below that for $\hat{\mathcal{C}}_{rmi}$.

The practical consequence of these bounds is that a single observed value, e.g.\ $\hat{\mathcal{C}}^D_{rcmi}=0.05$, is small or large only in comparison to its achievable upper bound given the observed marginals. In Sec.~\ref{sec:applications} we therefore report, for each regularized measure on each dataset, both the point estimate (with a bootstrap 95\% confidence interval) and the achievable upper bound obtained by maximising over deterministic couplings that preserve $p(x,z)$.

\subsection{Maxima of the regularized measures under fixed alphabet sizes alone}
\label{sec:alphabetmax}

We now answer the complementary question: how large can each regularized measure of direct correlation possibly become when \emph{only} the alphabet sizes $(d_X,d_Y,d_Z)$ are fixed and the maximization runs over \emph{all} joint probability distributions $p(x,y,z)$? For the regularized CMI, for example, we maximize $y=D_{JS}\bigl(p\|q\bigr)=H(m)-\tfrac12\bigl(H(p)+H(q)\bigr)$, with $m=(p+q)/2$ and $q$ the reconstruction (\ref{defqxyz}), and take $\sqrt{y}$. Because the PMI, ICMI and do-calculus reconstructions involve conditional distributions that are undefined on zero-probability cells, all suprema below are taken over strictly positive $p(x,y,z)$ (where every reconstruction is unambiguous); the suprema of the rPMI and rICMI are then approached, but not attained, along sequences converging to deterministic boundary couplings.

All closed forms below are built from a single function
\begin{equation}
\varphi(c) \;=\; H_2\Bigl(\frac{1+c}{2}\Bigr)\;-\;\frac{1}{2}H_2(c), \qquad c\in[0,1],
\label{defphi}
\end{equation}
where $H_2(u)=-u\log_2 u-(1-u)\log_2(1-u)$ is the binary entropy. The function $\varphi$ is strictly decreasing with $\varphi(0)=1$ and $\varphi(1)=0$; it is the JS divergence between a distribution and a second one that contains it with weight $c$,
as Lemma~\ref{lem:overlap} in the Appendix shows that $D_{JS}(P\|Q)=\varphi(c)$ whenever $Q= cP+(1-c)P^{\perp}$ with $P^{\perp}$ supported off the support of $P$. In particular $\varphi(1/d)$ coincides with the square of the bound (\ref{maxrmi}), $\hat{\mathcal{C}}_{rmi}^{\max}(d)=\sqrt{\varphi(1/d)}$.

\begin{proposition}[removal family and do family]
\label{th:rcmi}
Let $J^{*}(d_X,d_Y)$ denote the maximum of $D_{JS}(r(x,y)\|r(x)r(y))$ over all two-variable joint distributions $r$ on alphabets of sizes $d_X,d_Y$. Then, with $m=\min(d_X,d_Y)$ and for every $d_Z\ge 1$,
\begin{equation}
\max_p \hat{\mathcal{C}}_{rmi} \;=\; \max_p \hat{\mathcal{C}}^D_{rcmi} \;=\; \max_p \hat{\mathcal{C}}^D_{rmi,do} \;=\; \sqrt{J^{*}(d_X,d_Y)} ,
\end{equation}
and the three maxima are attained, e.g.\ at $p(x,y,z)=r^{*}(x,y)p(z)$ with $r^{*}$ a maximizer of $J^{*}$ and arbitrary $p(z)$ with full support. Moreover $J^{*}(d_X,d_Y)\ge\varphi(1/m)$, with equality verified numerically to within $10^{-9}$ for all $2\le d_X,d_Y\le 6$, the maximizer being the uniform deterministic coupling $r^{*}(x,y)=\delta_{x,y}/m$ supported on $m$ values of each variable.
\end{proposition}

Proposition~\ref{th:rcmi} (proved in the Appendix) states that neither the conditioning on $Z$ in the rCMI nor the do-intervention in $\hat{\mathcal{C}}^D_{rmi,do}$ enlarges the maximal value beyond that of the plain regularized mutual information: all three share the maximum $\sqrt{\varphi(1/m)}$, which is the bound (\ref{maxrmi}) evaluated at $d=m$, achieved by the perfectly correlated uniform pair $(X,Y)$ with $Z$ independent of it. In particular the maximum is independent of $d_Z$.

\begin{proposition}[regularized one-way ICMI]
\label{th:ricmi}
For all alphabet sizes with $d_Y\ge 2$,
\begin{equation}
\sup_p \hat{\mathcal{C}}^D_{ricmi}(X\rightarrow Y) \;\le\; \sqrt{1-\frac{1}{d_X}},
\label{icmibound}
\end{equation}
and the bound is approached (in the limit of deterministic couplings between $X$ and $Z$) whenever $d_X\le d_Z$. For $d_X>d_Z$ the supremum is strictly smaller and is given by a partition optimization (Appendix); for example
$\sup \hat{\mathcal{C}}^{D\,2}_{ricmi}(X{\to}Y)$ equals
$1/(2-\varphi(1/2))\approx 0.5922$ for $(d_X,d_Y,d_Z)=(3,3,2)$,
$\tfrac12+\tfrac12\varphi(1/2)\approx 0.6556$ for $(4,4,2)$, and
$\approx 0.7103$ for $(4,3,3)$.
The corresponding statements for $\hat{\mathcal{C}}^D_{ricmi}(X\leftarrow Y)$ follow by exchanging $d_X$ and $d_Y$. Numerically, the supremum of the two-way measure $\hat{\mathcal{C}}^D_{ricmi}(X{:}Y)$ coincides in all computed cases with the common value of the two one-way suprema, both being approached simultaneously; for $d_X=d_Y=d\le d_Z$ it equals $\sqrt{1-1/d}$.
\end{proposition}

\begin{proposition}[regularized PMI]
\label{prop:rpmi}
For all alphabet sizes with $d_Z\ge 2$,
\begin{equation}
\sup_p \hat{\mathcal{C}}^D_{rpmi} \;=\; \sqrt{\varphi\bigl(1/(d_X d_Y)\bigr)}
\label{pmibound}
\end{equation}
whenever $d_Z\ge d_X+d_Y-1$,
and the same value is already reached at $d_Z=d$ when $d_X=d_Y=d$ (in particular at $(2,2,2)$). The supremum is nondecreasing in $d_Z$, and for smaller $d_Z$ it is given by a weight optimization over stratified permutation couplings (Appendix A; the optimal weights equalize the per-stratum overlaps, as the worked $(3,3,2)$ case there shows); for example $\sup\hat{\mathcal{C}}^{D\,2}_{rpmi}=\max_{w}[\,w\varphi(w^2)+(1-w)\varphi((1-w)^2/2)\,]\approx 0.6480$ at $(3,3,2)$.
\end{proposition}

The suprema of the rPMI and the rICMI are approached along sequences of strictly positive distributions converging to \emph{sparse} deterministic couplings: for the rICMI, $X$ becomes a deterministic function of $Z$ while the conditional distribution of $Y$ jumps between disjoint supports; for the rPMI, the optimal configurations converge to ``generalized permutation'' couplings in which $X$, $Y$ and $Z$ are all deterministically linked (the simplest case being $X=Y=Z$ uniform for $d_X=d_Y=d_Z=2$). At the limiting distributions themselves the measures are defined only through the sparse-data strategies of Section~\ref{sec:sparse} --- and under the conservative strategy (c) they drop to $0$ there, while the CMI vanishes identically along these limits. The closer a dataset is to such a deterministic coupling, the more strongly the reported rPMI and rICMI values depend on the singularity strategy, which is a further practical reason to read each measure against its own maximal scale.

Table~\ref{tab:alphabetmax} collects the maxima for the symmetric case $d_X=d_Y=d_Z=d$. Three features deserve emphasis. First, the maxima are ordered,
\begin{equation}
\sqrt{\varphi(1/d)} \;<\; \sqrt{1-1/d} \;<\; \sqrt{\varphi(1/d^{2})} \;<\; 1 ,
\end{equation}
i.e.\ for every alphabet size the rCMI (and $\hat{\mathcal{C}}^D_{rmi,do}$) operates on the most compressed scale, the rICMI on an intermediate one, and the rPMI on the widest one, while NACE and RACE can reach the trivial bound $1$ exactly (Section~\ref{sec:upperbound}). A reported value of, say, $0.5$ therefore means ``$90\%$ of the maximum'' for a binary rCMI but only ``$68\%$ of the maximum'' for a binary rPMI. Second, every maximum tends to $1$ as $d\to\infty$, so the distinction matters most for the small alphabets typical of categorical data. Third, the maxima of the rCMI family are attained at benign, strictly positive distributions, whereas those of the rPMI/rICMI families live at sparse boundary configurations; this is the alphabet-size counterpart of the empirical observation in Section~\ref{sec:applications} that rPMI and rICMI tend to report systematically larger values than the rCMI on the same data.

\begin{table*}[htbp]
\centering
\caption{Maxima of the regularized measures over all joint distributions $p(x,y,z)$ with fixed alphabet sizes $d_X=d_Y=d_Z=d$ (measure scale, i.e.\ after the square root; logarithms base 2). The rCMI/rMI/do-MI column is exact and attained (Proposition~\ref{th:rcmi}, closed form $\sqrt{\varphi(1/d)}$, numerically confirmed to $10^{-9}$); the rICMI columns follow Proposition~\ref{th:ricmi} ($\sqrt{1-1/d}$, supremum); the rPMI column follows Proposition~\ref{prop:rpmi} ($\sqrt{\varphi(1/d^{2})}$, supremum, numerically confirmed to $10^{-7}$); NACE/RACE attain $1$ for every $d\ge 2$. }
\label{tab:alphabetmax}
\begin{tabular}{@{}lccccc@{}}
\hline
$d$ & $\hat{\mathcal{C}}_{rmi},\ \hat{\mathcal{C}}^D_{rcmi},\ \hat{\mathcal{C}}^D_{rmi,do}; $ & $\hat{\mathcal{C}}^D_{ricmi}(X{\to}Y); $ & $\hat{\mathcal{C}}^D_{ricmi}(X{:}Y);$ & $\hat{\mathcal{C}}^D_{rpmi};$ & $\mathcal{C}^D_{nace},\ \hat{\mathcal{C}}^D_{race}$ \\
\hline
2 & 0.5579 & 0.7071 & 0.7071 & 0.7408 & 1 \\
3 & 0.6776 & 0.8165 & 0.8165 & 0.8599 & 1 \\
4 & 0.7408 & 0.8660 & 0.8660 & 0.9102 & 1 \\
5 & 0.7810 & 0.8944 & 0.8944 & 0.9369 & 1 \\
$\to\infty$ & $\to 1$ & $\to 1$ & $\to 1$ & $\to 1$ & 1 \\
\hline
\end{tabular}
\end{table*}

For non-square alphabets the same machinery applies: the rCMI family depends only on $\min(d_X,d_Y)$; the one-way rICMI depends on $d_X$ (direction $X\to Y$) through the bound (\ref{icmibound}) and on $d_Z$ through the partition refinement; and the rPMI depends on the product $d_Xd_Y$ once $d_Z$ is large enough. All values quoted above were obtained independently by global numerical maximization (multi-start quasi-Newton/Adam optimization over the probability simplex in $64$-bit precision; see the Supplementary Material) and by the analytic constructions of the Appendix, which agree to the stated precision.

\section{Comparison of measures of direct correlation with examples}
\label{sec:dm}

In this section, we compare the different measures of direct correlation with examples of fixed models.

\subsection{A decision-making model}
Now we consider a decision-making model with three participants Xie, Yu, and Zhang, who are going to vote on a particular proposal (see Fig.~\ref{fig_exampledecisionmaking}). The variables X, Y, and Z respectively denote their choices, with each value either $0$ (veto) or $1$ (vote).
\begin{figure}[htbp]
\centering
\includegraphics[width=0.17\textwidth]{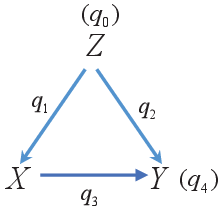}
\caption{A decision-making model. Here $q_1$ denotes the strength of Z's influence on X, similarly for $q_2$ and $q_3$.}
\label{fig_exampledecisionmaking}\hspace{-2mm}
\end{figure}
Zhang makes the choice $Z=z$ with probability $p(z)$ given by
\begin{equation}
p(z) =\left\{
\begin{aligned}
\frac{1+q_0}{2}  &\;\;\; (z=0) \\
\frac{1-q_0}{2}  &\;\;\; (z=1) .
\end{aligned}
\right.
\end{equation}
Xie usually seeks advice from Zhang, and let $q_1$ denote Zhang's influence strength on Xie's choice. $q_1=1$ means a complete influence and Xie's choice is the same as Zhang's ($x=z$); $q_1=0$ means Xie's choice is independent of Zhang's; and $q_1=-1$ means Xie's choice is always different from Zhang's. The dependence of Xie's choice on Zhang's is given by the conditional probability
\begin{equation}
p(x|z) =\left\{
\begin{aligned}
\frac{1+q_1}{2}  &\;\;\; (x=z) \\
\frac{1-q_1}{2}  &\;\;\; (x=1-z) .
\end{aligned}
\right.
\end{equation}
Yu's choice depends both on Zhang's via a strength $q_2$ and on Xie's via a strength $q_3$ in the following way. Yu first makes a pre-choice $y_1$ with
\begin{equation}
p(y_1|x) =\left\{
\begin{aligned}
\frac{1+q_3}{2}  &\;\;\; (y_1=x) \\
\frac{1-q_3}{2}  &\;\;\; (y_1=1-x) .
\end{aligned}
\right.
\end{equation}
We see that $y_1$ is independent of $x$ when $q_3=0$, and $y_1$ is more likely to be the same as $x$ when $q_3$ approaches $1$, be the opposite of $x$ when $q_3$ approaches $-1$.
Similarly Yu also takes advice from Zhang by making another pre-choice $y_2$ with
\begin{equation}
p(y_2|z) =\left\{
\begin{aligned}
\frac{1+q_2}{2}  &\;\;\; (y_2=z) \\
\frac{1-q_2}{2}  &\;\;\; (y_2=1-z) .
\end{aligned}
\right.
\end{equation}
When Zhang and Xie's advice agrees, i.e., Yu's pre-choices $y_1$ and $y_2$ are equal, then Yu follows, i.e., $y=y_1=y_2$; otherwise, when $y_1 \neq y_2$, then Yu makes a choice $Y=y$ according to his own judgement with $p(y=0) = \frac{1+q_4}{2}$ and $p(y=1) = \frac{1-q_4}{2}$. Here $q_4$ denotes Yu's own bias towards the choice $y=0$. In other words,
\begin{align}
p(y|y_1,y_2) =\left\{
\begin{aligned}
1  &\;\;\; (y=y_1) \\
0  &\;\;\; (y=1- y_1)
\end{aligned}
when \;\; y_1=y_2
\right. \\
p(y|y_1,y_2) =\left\{
\begin{aligned}
\frac{1+q_4}{2}  &\;\;\; (y =0) \\
\frac{1-q_4}{2}  &\;\;\; (y =1)  .
\end{aligned}
when \;\; y_1 \neq y_2
\right.
\end{align}

From the above conditional probability distributions, we have the joint probability distribution of all variables
$p(y,y_1,y_2,x,z)=p(y|y_1,y_2) p(y_1|x) p(y_2|z) p(x|z) p(z)$ and that of three variables X, Y, and Z
\begin{equation}
p(y,x,z)=\sum_{y_1,y_2} p(y|y_1,y_2) p(y_1|x) p(y_2|z) p(x|z) p(z) .
\end{equation}
From the joint probability distribution $p(y,x,z)$ we can analyze all correlations between the variables. In particular, we focus on direct correlation from X to Y.

This decision-making model has pre-assumed influence relationships, in particular, $q_1$ represents the influence strength from Z to X, $q_2$ represents the influence strength from Z to Y, $q_3$ represents the influence strength from X to Y. The different measures of direct correlation discussed have different values in general, but each gives a faithful measure of the direct correlation. With all parameters but $q_3$ fixed, we find that the measures of direct correlation increase with the increase of $q_3$ (see Fig. \ref{fig_DMmeasDCvq3}).
\begin{figure}[htbp]
\centering
\includegraphics[width=0.5\textwidth]{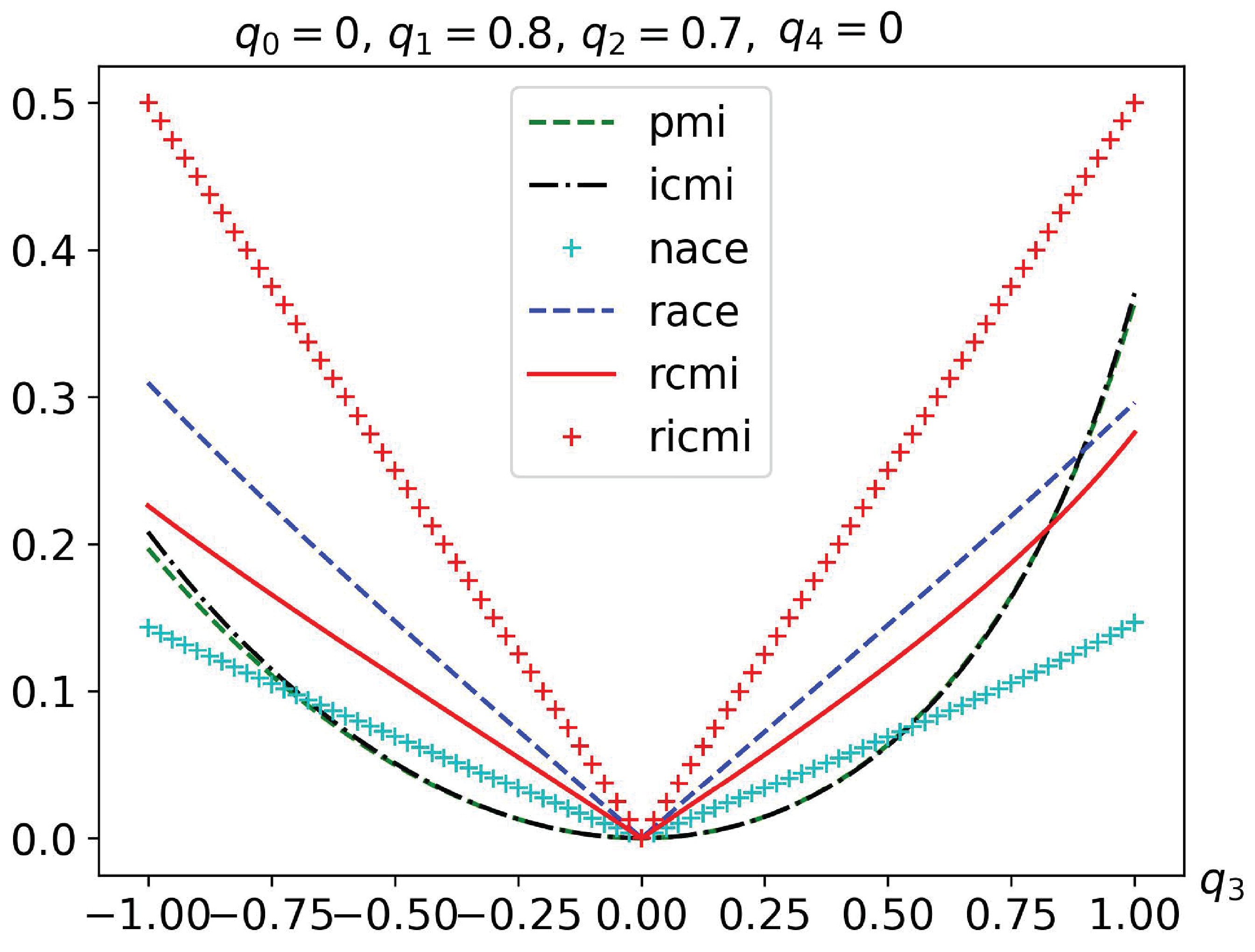}
\caption{Various measures of direct correlation in the decision-making model as functions of $q_3$ with other parameters fixed.}
\label{fig_DMmeasDCvq3}\hspace{-2mm}
\end{figure}

Consider the special case $q_0=0$, $q_1=q_2=q_3=1$, $q_4=0$: then $x=z$, $y_1=x=z$, $y_2=z$, so $y=x=z$ always, and the joint distribution reduces to $p(x,y,z)=1/2$ when $x=y=z$ and $0$ otherwise. The CMI vanishes, while the other measures depend on the singularity strategy because the data are sparse. Under strategy (c) all direct-correlation measures above vanish. Under strategy (a) or (b) several are non-zero, e.g.\ $\mathcal{C}^D_{pmi}\approx 0.83$, $\mathcal{C}^D_{nace}=0.50$, $\hat{\mathcal{C}}^D_{race}\approx 0.43$, $\mathcal{C}^D_{mi,\mathrm{do}}(X\to Y)\approx 0.19$. Strategy (c) correctly refuses to manufacture direct correlation here: this sparse case is indistinguishable from the three alternative influence structures in Fig.~\ref{fig_DMspecialequivalentcases}, all of which have zero direct correlation. This deterministic configuration $X=Y=Z$ is also precisely the type of boundary coupling along which the regularized PMI and ICMI approach their alphabet-size maxima of Section~\ref{sec:alphabetmax}, which explains the extreme sensitivity of these two measures to the singularity strategy in its vicinity.
\begin{figure}[htbp]
\centering
\includegraphics[width=0.4\textwidth]{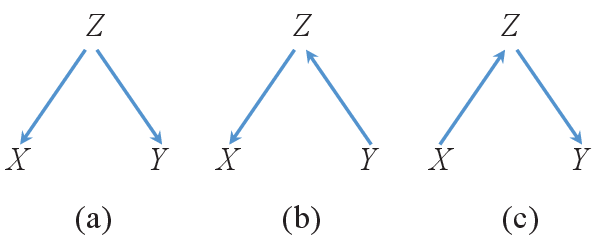}
\caption{Three models of different relationships indistinguishable to the special case of decision-making model when $q_0=0$, $q_1=q_2=q_3=1$, and $q_4=0$.}
\label{fig_DMspecialequivalentcases}\hspace{-2mm}
\end{figure}

\subsection{A simpler decision-making model}
Now we consider a simpler version of the decision-making model with three participants Xie, Yu, and Zhang, who are going to vote on a particular proposal (see Fig.~\ref{fig_exampledecisionmaking_1}). The variables X, Y, and Z respectively denote their choices, with each value either $0$ (veto) or $1$ (vote).
\begin{figure}[htbp]
\centering
\includegraphics[width=0.15\textwidth]{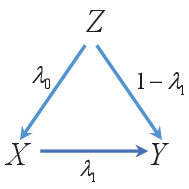}
\caption{A simplified decision-making model. $\lambda_0$ is the influence strength from $Z$ to $X$ and $\lambda_1$ parametrises the influence from $X$ to $Y$.}
\label{fig_exampledecisionmaking_1}\hspace{-2mm}
\end{figure}
Zhang makes an unbiased choice $Z=z$ with an equal probability $p(z)=1/2$.
Xie usually seeks advice from Zhang, and let $\lambda_0$ denote Zhang's influence strength on Xie's choice. $\lambda_0=1$ means a complete influence and Xie's choice is always the same as Zhang's ($x=z$); $\lambda_0=0$ means Xie's choice is independent of Zhang's; and $\lambda_0=-1$ means Xie's choice is always different from Zhang's. The dependence of Xie's choice on Zhang's is given by the conditional probability
\begin{equation}
p(x|z) =\left\{
\begin{aligned}
\frac{1+\lambda_0}{2}  &\;\;\; (\textrm{if }x=z) \\
\frac{1-\lambda_0}{2}  &\;\;\; (\textrm{if }x=1-z) .
\end{aligned}
\right.
\end{equation}
Yu's choice depends both on Zhang's via a strength $1-\lambda_1$ and on Xie's via a strength $\lambda_1$ in the following way.
If Zhang and Xie's choices are the same ($x=z$), then Yu also chooses the same value; otherwise ($x\neq z$) he chooses $y=x$ with probability $\lambda_1$ and $y=z$ with probability $1- \lambda_1$.
Therefore, one can easily write down the conditional probability $p(y|x,z)$. From $p(x,y,z)=p(y|x,z)p(z)$ the joint probability distribution $p(x,y,z)$ can be easily obtained
\begin{equation}
p(x,y,z) =\left\{
\begin{array}{l}
\frac{1+\lambda_0}{4}  \;\;\;\;\; (\textrm{if }x=y=z) \\
\;\; 0 \;\;\;\;\;\;\;\;\; (\textrm{if }x=z, \; y=1-z) \\
\frac{(1-\lambda_0)\lambda_1}{4}  \;  (\textrm{if }x=1-z, \; y=1-z) \\
\frac{(1-\lambda_0)(1-\lambda_1)}{4}  \; (\textrm{if }x=1-z, \; y=z) .
\end{array}
\right.
\end{equation}
From the joint probability distribution $p(y,x,z)$ we can analyze all correlation measures between the variables. In particular, we focus on direct correlation from X to Y.

This decision-making model has a pre-assumed influence structure: $X$ can only be influenced by $Z$ while $Y$ can be influenced by both $X$ and $Z$.
In particular, $\lambda_0$ represents the influence strength from Z to X. The case $\lambda_0=0$ indicates that Z has no influence on X (i.e., X is independent of Z), the case $\lambda_0=\pm 1$ indicates that Z has full influence on X with $X=Z$ ($\lambda_0= 1$) or $X=1-Z$ ($\lambda_0= -1$). Here $\lambda_1$ represents the direct influence strength from X to Y. The case $\lambda_1 =0$ means that X has no influence on Y (while Z has full influence on Y), and the case $\lambda_1 =1$ means that X has full influence on Y (while Z has no influence on Y).
The different measures of direct correlation discussed have different values in general, but each gives a faithful measure of the direct correlation.
We fix $\lambda_0$ and plot each direct-correlation measure as a function of $\lambda_1$ in Fig.~\ref{fig_DMSmeasdirclambda0fixed}. Every measure is monotone non-decreasing in $\lambda_1$, providing a sanity check on the response of the regularised family to increasing direct dependence. Compared to NACE and RACE, $\hat{\mathcal{C}}^D_{\mathrm{rcmi}}$ and $\hat{\mathcal{C}}^D_{\mathrm{ricmi}}$ are more sensitive at small $\lambda_1$ and saturate more gently as $\lambda_1\to 1$. This follows directly from the definitions: NACE and RACE depend only on the most discriminating pair of $X$-values, whereas $\hat{\mathcal{C}}^D_{\mathrm{rcmi}}$ and $\hat{\mathcal{C}}^D_{\mathrm{ricmi}}$ integrate over all values of $X$ and $Y$ and are therefore smoother functions of the model parameters.

\begin{figure}[htbp]
\centering
\includegraphics[width=0.4\textwidth]{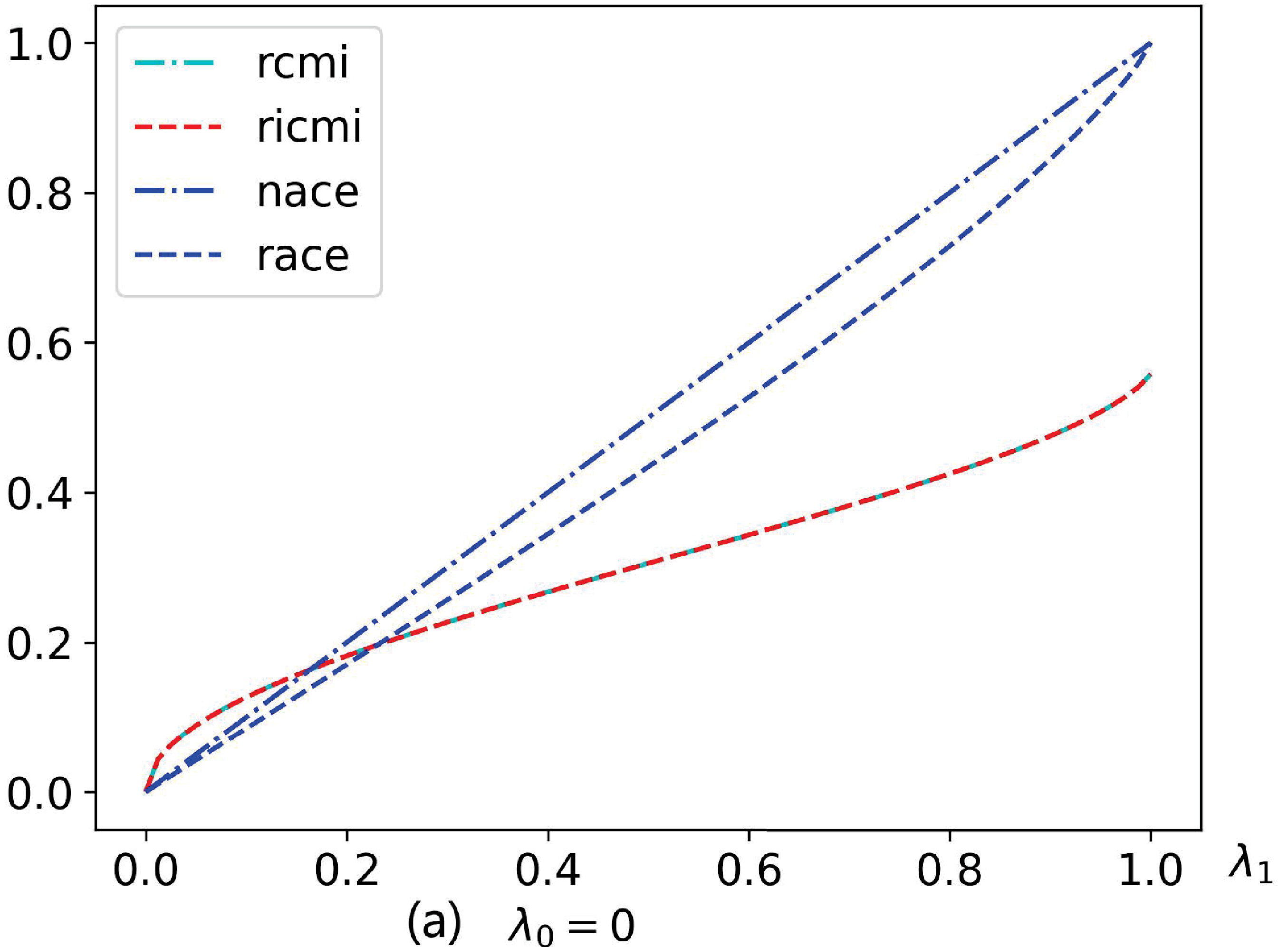}
\includegraphics[width=0.4\textwidth]{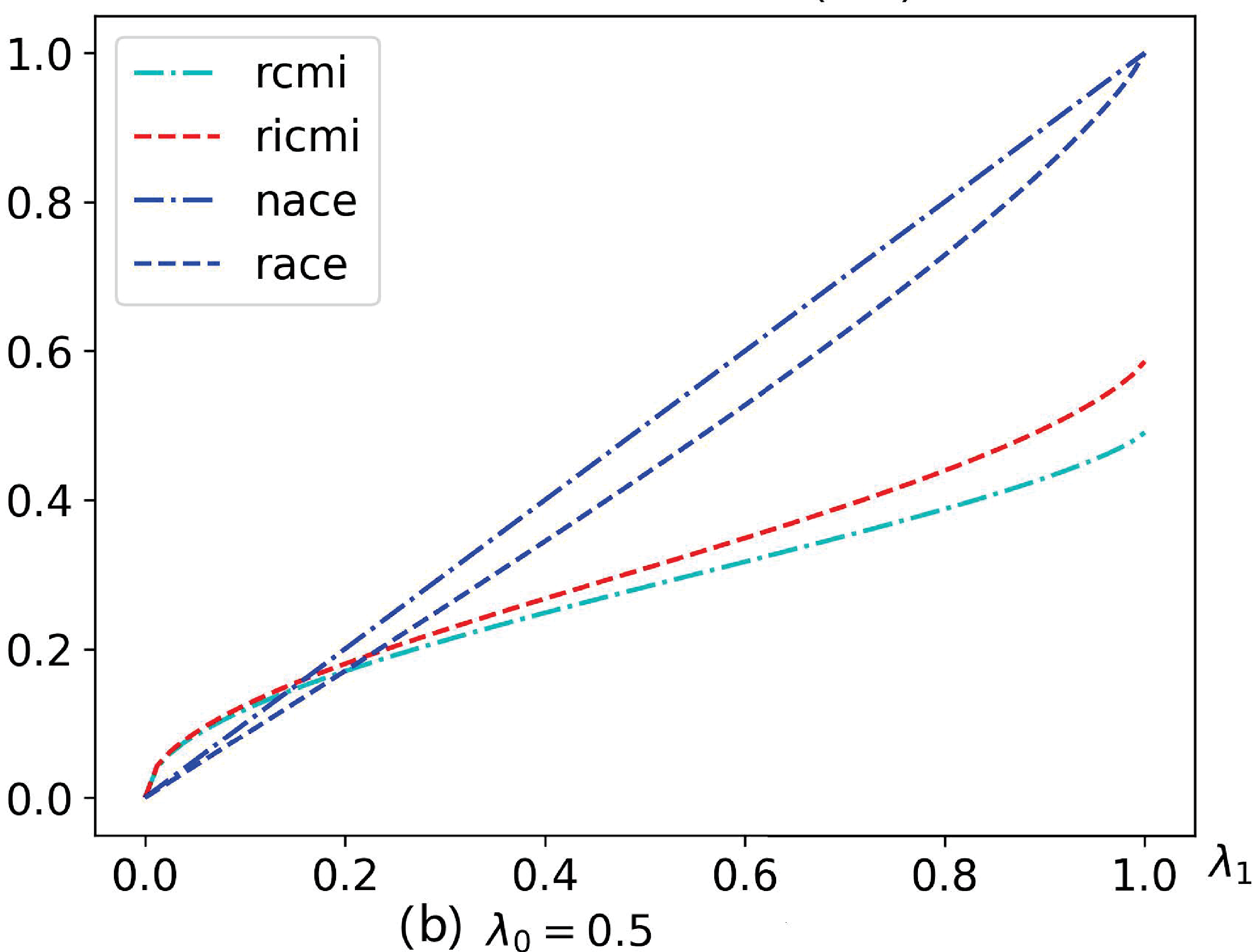}
\includegraphics[width=0.4\textwidth]{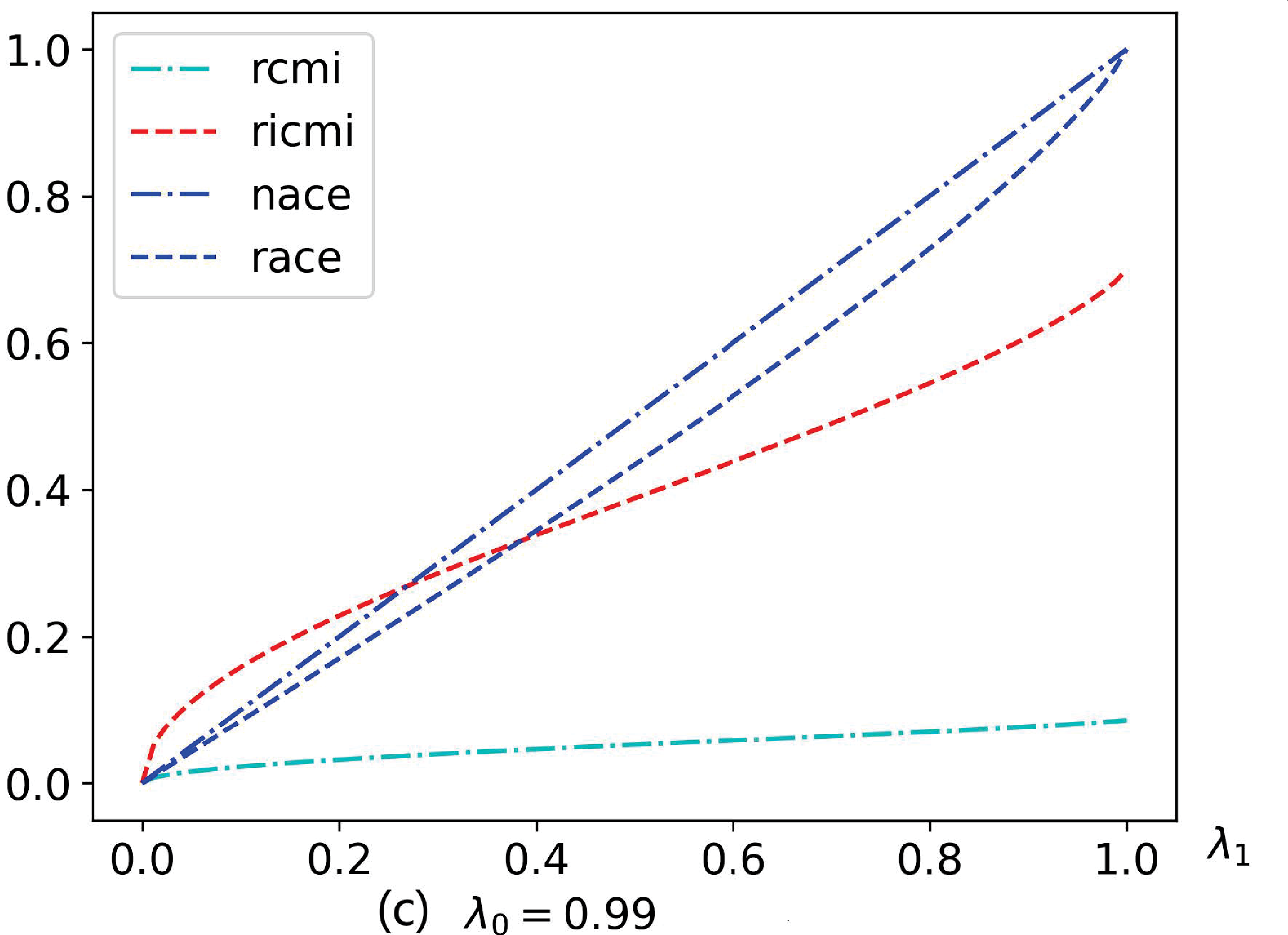}
\caption{Direct-correlation measures in the simplified decision-making model, plotted as functions of $\lambda_1$ at three fixed values of $\lambda_0$: (a)~$\lambda_0=0$, (b)~$\lambda_0=0.5$, (c)~$\lambda_0=0.99$.}
\label{fig_DMSmeasdirclambda0fixed}\hspace{-2mm}
\end{figure}

\section{Applications to real-world data}
\label{sec:applications}

We now apply the measures to three publicly available datasets chosen to span a range of direct-correlation magnitudes: the Titanic 1912 passenger record~\cite{titanic_kaggle}, the UCI Adult (Census Income) dataset~\cite{ucimladult}, and the Bickel \emph{et al.}\ Berkeley 1973 admissions data~\cite{bickel1975sex}. All point estimates are accompanied by bootstrap~\cite{Efron1979bootstrap} $95\%$ confidence intervals (CI) obtained from $B=1000$ resamples of the raw observation-level data. For every regularized direct-correlation measure we additionally report the \emph{achievable upper bound}, obtained by exhaustively enumerating the deterministic couplings $Y=f(X,Z)$ compatible with the observed marginal $p(x,z)$ and keeping the maximum (see Sec.~\ref{sec:upperbound}). Rare zero cells are treated with strategy~(b) of Sec.~\ref{sec:sparse}; for the three datasets used here, all cells of the empirical joint distribution are non-empty, so the singularity strategy is not activated in any of the reported numerical values.

\subsection{Dataset I: Titanic passenger survival (1912)}
\label{sec:titanic}

The Titanic training split~\cite{titanic_kaggle}, widely used as a pedagogical benchmark in statistics and machine learning, records $891$ passengers, of whom $342$ survived ($38.4\%$). We take $X$ as passenger class (Pclass $\in\{1,2,3\}$), $Y$ as survival (0 or 1), and $Z$ as sex (0=female, 1=male). Sex is a strong predictor of survival (``women and children first''), and passenger class is associated with survival through a combination of physical factors (cabin location, access to lifeboats), crew allocation, and other class-correlated social factors; disentangling the direct class--survival association from the sex--survival association is the quantitative question addressed below. The passenger data is given in Table~\ref{tab:titanic_counts}.

\begin{table}[htbp]
\centering
\caption{Titanic passenger data: survival counts (survivors/total) cross-classified by class and sex. Survival rate in parentheses.}
\label{tab:titanic_counts}
\begin{tabular}{c|cc|c}
\hline
Pclass & Female & Male & Aggregate \\
\hline
1 & $91/94$ $(96.8\%)$   & $45/122$ $(36.9\%)$ & $136/216$ \\
2 & $70/76$ $(92.1\%)$   & $17/108$ $(15.7\%)$ & $87/184$ \\
3 & $72/144$ $(50.0\%)$  & $47/347$ $(13.5\%)$ & $119/491$ \\
\hline
All & $233/314$ $(74.2\%)$ & $109/577$ $(18.9\%)$ & $342/891$ \\
\hline
\end{tabular}
\end{table}

The marginal Pclass--Survival correlation is strong ($\mathcal{C}_{\mathrm{pcc}}=-0.339$, $\hat{\mathcal{C}}_{\mathrm{rmi}}=0.146$). The partial correlation after conditioning on sex is $\mathcal{C}^D_{\mathrm{pc}}=-0.321$, essentially unchanged, indicating that sex does not explain away the class effect. All nonlinear direct-correlation measures sit in the $0.12$--$0.22$ range (Table~\ref{tab:realmeasures}, left block). Interestingly, $\hat{\mathcal{C}}^D_{\mathrm{rcmi}}=0.160$ \emph{exceeds} $\hat{\mathcal{C}}_{\mathrm{rmi}}=0.146$; a plausible reason is that sex is itself strongly associated with survival, and stratifying on sex removes part of the total correlation that is driven by the sex--survival path while leaving the direct class--survival effect intact. The same ordering shows up in $\hat{\mathcal{C}}^D_{\mathrm{rpmi}}=0.195$ and $\hat{\mathcal{C}}^D_{\mathrm{ricmi}}=0.190$. The two directional components of the regularised ICMI differ sizeably: $\hat{\mathcal{C}}^D_{\mathrm{ricmi}}(X\leftarrow Y)=0.221$ exceeds $\hat{\mathcal{C}}^D_{\mathrm{ricmi}}(X\to Y)=0.159$ by roughly $40\%$. This asymmetry is a property of the ICMI construction: the two reconstructions are normalised with respect to different marginal distributions, so even in the absence of any meaningful causal asymmetry the two one-way values generally differ; the achievable upper bounds in Table~\ref{tab:realmeasures} (0.326 vs.\ 0.252) show that the $X\leftarrow Y$ direction simply has more room to grow for this dataset. Finally, $\mathcal{C}^D_{\mathrm{nace}}=0.316$ and $\hat{\mathcal{C}}^D_{\mathrm{race}}=0.275$, the largest among the direct-correlation measures: both capture the contrast between the intervened survival probabilities $p(Y{=}1|\mathrm{do}(\mathrm{Pclass}{=}1))\approx 0.58$ and $p(Y{=}1|\mathrm{do}(\mathrm{Pclass}{=}3))\approx 0.26$ (computed under strategy (b) using the observed $p(\mathrm{sex})$).

\subsection{Dataset II: UCI Adult (Census Income)}
\label{sec:adult}

The UCI Adult dataset~\cite{ucimladult} contains $n=32561$ records extracted from the 1994 US Census, with income dichotomised at \$50\,000/year. We take $X$ as the education level binned into four ordinal groups (0: at most 12th grade, 1: high-school graduate or some college, 2: Associate's or Bachelor's, 3: Master's, professional or doctoral), $Y$ as the binary indicator of income above \$50\,000, and $Z$ as sex (0=female, 1=male). The marginal proportions of $X$ are $0.132$, $0.540$, $0.244$, $0.084$; all $4\times 2\times 2 = 16$ cells of the empirical joint distribution are non-empty (minimum cell count $23$), so no singularity strategy is invoked.

The direct correlation between education and income is strong, both marginally ($\mathcal{C}_{\mathrm{pcc}}=+0.346$, $\hat{\mathcal{C}}_{\mathrm{rmi}}=0.148$) and conditionally ($\mathcal{C}^D_{\mathrm{pc}}=+0.349$, $\hat{\mathcal{C}}^D_{\mathrm{rcmi}}=0.149$). This is a regime in which the marginal and the direct correlation agree: education exerts a genuine direct effect on income and the sex confounder, though non-trivial, does not substantially bias the marginal picture. The very narrow bootstrap confidence intervals (widths $\lesssim 0.01$, thanks to the large sample size) make this dataset a useful stability test for the measures. Beyond the linear/entropic measures, the pairwise measures NACE $=0.544$ and RACE $=0.521$ are very large: the highest and lowest education buckets differ sharply in their intervened income distributions, and both NACE and RACE are by construction sensitive to the most discriminating pair of $X$ values.

\subsection{Dataset III: Berkeley 1973 admissions --- weak direct correlation}
\label{sec:berkeley}

As a counterpoint to the first two datasets, we revisit the Bickel \emph{et al.}~\cite{bickel1975sex} 1973 UC Berkeley admissions data ($n=4526$, six largest graduate departments) discussed in Sec.~I. We take $X$ as sex (0=F, 1=M), $Y$ as admission (0=rejected, 1=admitted), and $Z$ as department ($6$ values). The contingency is tabulated in Table~\ref{tab:berkeley_counts}.

\begin{table}[htbp]
\centering
\caption{Bickel \emph{et al.}~\cite{bickel1975sex} cross-tabulation of the 1973 Berkeley admissions by sex and department; admit rate (\%) in parentheses.}
\label{tab:berkeley_counts}
\begin{tabular}{c|cc|c}
\hline
Dept & Male & Female & Aggregate \\
\hline
A     & $512/825$\ $(62.1)$ & $89/108$\ $(82.4)$  & $601/933$ \\
B     & $353/560$\ $(63.0)$ & $17/25$\ $(68.0)$   & $370/585$ \\
C     & $120/325$\ $(36.9)$ & $202/593$\ $(34.1)$ & $322/918$ \\
D     & $138/417$\ $(33.1)$ & $131/375$\ $(34.9)$ & $269/792$ \\
E     & $53/191$\ $(27.7)$  & $94/393$\ $(23.9)$  & $147/584$ \\
F     & $22/373$\ $(5.9)$   & $24/341$\ $(7.0)$   & $46/714$ \\
\hline
Total & $1198/2691$\ $(44.5)$ & $557/1835$\ $(30.4)$ & $1755/4526$ \\
\hline
\end{tabular}
\end{table}

The marginal admit rate favours male applicants by $14$ points, whereas within every department the rates are close or even favour women in four of the six departments. The marginal correlation is therefore non-negligible ($\mathcal{C}_{\mathrm{pcc}}=+0.143$, $\hat{\mathcal{C}}_{\mathrm{rmi}}=0.061$), while every direct-correlation measure drops to the $0.02$--$0.05$ range (Table~\ref{tab:realmeasures}, right block). In particular, $\hat{\mathcal{C}}^D_{\mathrm{rcmi}}=0.030$ is less than half of $\hat{\mathcal{C}}_{\mathrm{rmi}}$, and $\hat{\mathcal{C}}^D_{\mathrm{rmi,do}}=0.018$ less than one third. This quantifies, on the regularised scale, the classical qualitative observation that most of the marginal sex--admission correlation is inherited from the department, not from a direct sex effect. The two one-way ICMI components differ ($\hat{\mathcal{C}}^D_{\mathrm{ricmi}}(X\to Y)=0.053$ vs.\ $\hat{\mathcal{C}}^D_{\mathrm{ricmi}}(X\leftarrow Y)=0.037$); this asymmetry is a property of the two-sided ICMI construction (different marginals enter the two reconstructions) and should not be over-interpreted as a causal-direction finding, especially given that sex is an exogenous pre-treatment variable in this dataset.

\begin{table*}[htbp]
\centering
\caption{Correlation and direct-correlation measures on the three real datasets. For every regularized direct-correlation measure we report the point estimate together with a bootstrap $95\%$ confidence interval (in square brackets) and the achievable upper bound $\hat{\mathcal{C}}^{\max}$ obtained by enumerating all deterministic couplings $Y=f(X,Z)$ compatible with the observed $p(x,z)$. Values below the thick horizontal line are direct-correlation measures. CIs computed from $B=1000$ bootstrap resamples (seed 20260419).}
\label{tab:realmeasures}
\begin{tabular}{l|lll|lll|lll}
\hline
 & \multicolumn{3}{c|}{Titanic ($n=891$)} & \multicolumn{3}{c|}{Adult ($n=32561$)} & \multicolumn{3}{c}{Berkeley ($n=4526$)} \\
Measure & value & 95\% CI & $\hat{\mathcal{C}}^{\max}$ & value & 95\% CI & $\hat{\mathcal{C}}^{\max}$ & value & 95\% CI & $\hat{\mathcal{C}}^{\max}$ \\
\hline
$\mathcal{C}_{\mathrm{pcc}}$                    & $-0.339$ & $[-0.40,-0.28]$ & ---   & $+0.346$ & $[+0.336,+0.356]$ & ---   & $+0.143$ & $[+0.113,+0.172]$ & ---   \\
$\mathcal{C}^D_{\mathrm{pc}}$                   & $-0.321$ & $[-0.38,-0.26]$ & ---   & $+0.349$ & $[+0.338,+0.359]$ & ---   & ---      & ---               & ---   \\
$\hat{\mathcal{C}}_{\mathrm{rmi}}$              & $0.146$  & $[0.119,0.174]$ & $0.555$ & $0.148$  & $[0.143,0.152]$  & $0.556$ & $0.061$  & $[0.048,0.074]$  & $0.549$ \\
\hline
$\hat{\mathcal{C}}^D_{\mathrm{rcmi}}$           & $0.160$  & $[0.129,0.186]$ & $0.246$ & $0.149$  & $[0.144,0.154]$  & $0.247$ & $0.030$  & $[0.021,0.046]$  & $0.222$ \\
$\hat{\mathcal{C}}^D_{\mathrm{rpmi}}$           & $0.195$  & $[0.165,0.224]$ & $0.206$ & $0.152$  & $[0.148,0.157]$  & $0.190$ & $0.042$  & $[0.030,0.064]$  & $0.277$ \\
$\hat{\mathcal{C}}^D_{\mathrm{ricmi}}(X{\to}Y)$ & $0.159$  & $[0.130,0.186]$ & $0.252$ & $0.148$  & $[0.143,0.152]$  & $0.248$ & $0.053$  & $[0.037,0.082]$  & $0.304$ \\
$\hat{\mathcal{C}}^D_{\mathrm{ricmi}}(X{\leftarrow}Y)$ & $0.221$  & $[0.173,0.257]$ & $0.326$ & $0.156$  & $[0.151,0.162]$  & $0.392$ & $0.037$  & $[0.026,0.059]$  & $0.343$ \\
$\hat{\mathcal{C}}^D_{\mathrm{ricmi}}$          & $0.190$  & $[0.154,0.221]$ & $0.269$ & $0.152$  & $[0.147,0.157]$  & $0.284$ & $0.045$  & $[0.031,0.070]$  & $0.310$ \\
$\mathcal{C}^D_{\mathrm{nace}}$                 & $0.316$  & $[0.247,0.383]$ & $1.000$ & $0.544$  & $[0.525,0.562]$  & $1.000$ & $0.043$  & $[0.008,0.078]$  & $1.000$ \\
$\hat{\mathcal{C}}^D_{\mathrm{race}}$           & $0.275$  & $[0.215,0.332]$ & $1.000$ & $0.521$  & $[0.505,0.538]$  & $1.000$ & $0.037$  & $[0.007,0.067]$  & $1.000$ \\
$\hat{\mathcal{C}}^D_{\mathrm{rmi,do}}$         & $0.116$  & $[0.091,0.141]$ & $0.555$ & $0.144$  & $[0.139,0.148]$  & $0.556$ & $0.018$  & $[0.003,0.033]$  & $0.549$ \\
\hline
\end{tabular}
\end{table*}

\subsection{Cross-dataset comparison and complementary strengths}
\label{sec:crossdataset}

Figure~\ref{fig_realdata} summarises the seven regularized direct-correlation measures of Table~\ref{tab:realmeasures} across the three datasets, with the bootstrap CI as error bar and the achievable upper bound as a dashed cap. The complementary behaviour of the measures becomes apparent:
\begin{itemize}
\item $\hat{\mathcal{C}}^D_{\mathrm{rcmi}}$ is the smoothest and most conservative removal-of-direct-correlation measure: it integrates the discrepancy between $p(x,y,z)$ and $p(x|z)p(y|z)p(z)$ and therefore reports the overall magnitude of the direct correlation.
\item $\hat{\mathcal{C}}^D_{\mathrm{rpmi}}$ and $\hat{\mathcal{C}}^D_{\mathrm{ricmi}}$ are larger than $\hat{\mathcal{C}}^D_{\mathrm{rcmi}}$ on Titanic (where the PMI/ICMI factorisation is more aggressive) but comparable on Adult (where the confounder is weak relative to the direct effect). The directional $\hat{\mathcal{C}}^D_{\mathrm{ricmi}}$ components additionally expose the normalisation asymmetry discussed in Sec.~\ref{sec:titanic}.
    This ordering is consistent with the alphabet-size maxima of Section~\ref{sec:alphabetmax}: the rPMI and rICMI operate on intrinsically wider scales than the rCMI.
\item $\mathcal{C}^D_{\mathrm{nace}}$ and $\hat{\mathcal{C}}^D_{\mathrm{race}}$ saturate to very large values on Adult, where a specific contrast (highest vs.\ lowest education bucket) drives most of the effect. They are smaller than the other direct measures on Berkeley, where no single pair of $X$ values dominates.
\item $\hat{\mathcal{C}}^D_{\mathrm{rmi,do}}$ is the most conservative of all measures: it applies the mutual-information construction to the intervened joint $p_{\mathrm{do}}(x,y)$ and quantifies the $X$--$Y$ correlation that survives an intervention on $X$.
\end{itemize}
In every case the distance of the observed value from its achievable upper bound (dashed cap in Fig.~\ref{fig_realdata}) is the honest indicator of ``how correlated is correlated''; the ratio ${\hat{\mathcal{C}}}/{\hat{\mathcal{C}}^{\max}}$ is $\sim 65\%$ on Titanic, $\sim 60\%$ on Adult and $\sim 14\%$ on Berkeley for $\hat{\mathcal{C}}^D_{\mathrm{rcmi}}$, giving a quantitative confirmation of the qualitative comparison in the three subsections above.

\begin{figure*}[htbp]
\centering
\includegraphics[width=0.95\textwidth]{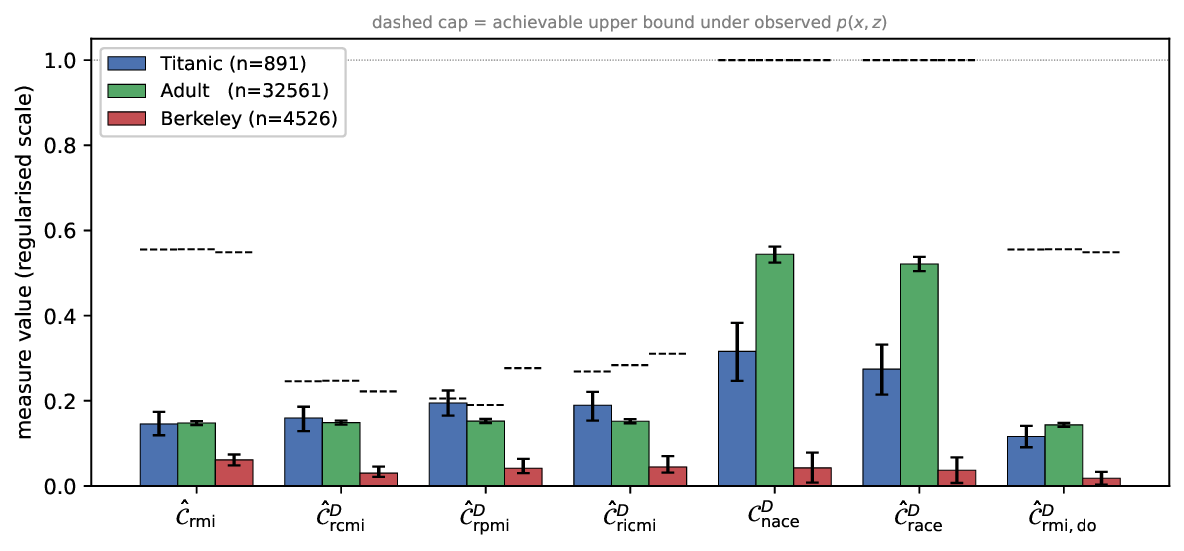}
\caption{Regularized direct-correlation measures on Titanic ($n=891$), Adult ($n=32561$) and Berkeley ($n=4526$). Error bars are bootstrap $95\%$ confidence intervals from $B=1000$ resamples (seed 20260419); dashed caps are the achievable upper bounds $\hat{\mathcal{C}}^{\max}$ obtained by enumerating all deterministic couplings $Y=f(X,Z)$ that preserve the observed $p(x,z)$. The dotted grey line marks the trivial upper bound $1$.}
\label{fig_realdata}
\end{figure*}

\section{Conclusion}

For a joint distribution $p(x,y,z)$ we have organised the existing measures of the direct correlation between $X$ and $Y$ into two families: the first \emph{removes} the direct correlation from the distribution and measures the induced shift (CMI, PMI, ICMI), while the second relies on Pearl's do-calculus and measures the response of the distribution of $Y$ to an intervention on $X$ (ACE, NACE, do-based mutual information). For every Kullback-Leibler-based member of either family we have introduced a Jensen-Shannon-based regularised analogue; the resulting measures are bounded by $1$, free of the singularities that plague sparse-data applications of the KL divergence, and equipped with an explicit achievable upper bound that depends on the observed marginals and the alphabet sizes. The JS metrization itself is not new~\cite{Lin1991,endres2003metric}; our contribution is to organise the existing direct-correlation measures into a single coherent family built on $\sqrt{D_{JS}}$, to clarify the relationships between their upper bounds, and to compare them empirically with bootstrap confidence intervals.
In particular, we have derived the maximal value of each regularized measure when only the alphabet sizes are fixed: the regularized CMI and the do-based regularized mutual information share the maximum $\sqrt{\varphi(1/\min(d_X,d_Y))}$ of the regularized mutual information, the one-way regularized ICMI is capped by $\sqrt{1-1/d_X}$, and the regularized PMI reaches $\sqrt{\varphi(1/(d_Xd_Y))}$; these maxima provide the natural scale against which any reported value should be read, and they expose the boundary configurations near which the PMI- and ICMI-type measures become sensitive to the sparse-data strategy.
On a decision-making toy model and on three public real datasets (Titanic 1912, UCI Adult 1994, and UC Berkeley 1973) the regularised measures behave consistently: they are substantial on the two datasets with a strong direct effect and markedly below the regularised total MI on the Berkeley data, where most of the marginal correlation is inherited from the department confounder. The choice among the measures is determined by the aspect of the direct correlation one wishes to emphasise --- overall magnitude ($\hat{\mathcal{C}}^D_{\mathrm{rcmi}}$), direction-resolved magnitude ($\hat{\mathcal{C}}^D_{\mathrm{ricmi}}$), worst-case pairwise contrast ($\mathcal{C}^D_{\mathrm{nace}}$, $\hat{\mathcal{C}}^D_{\mathrm{race}}$), or post-intervention correlation ($\hat{\mathcal{C}}^D_{\mathrm{rmi,do}}$) --- and the achievable upper bound provides the correct scale against which each observed value should be read.

\section*{Data availability}
All three real datasets used in Sec.~\ref{sec:applications} are publicly available. The Titanic passenger record~\cite{titanic_kaggle} is mirrored at \url{https://raw.githubusercontent.com/datasciencedojo/datasets/master/titanic.csv}. The UCI Adult (Census Income) dataset~\cite{ucimladult} is available at \url{https://archive.ics.uci.edu/ml/machine-learning-databases/adult/adult.data}. The Bickel~\emph{et al.}\ 1973 Berkeley admissions table~\cite{bickel1975sex} is bundled with the R \texttt{datasets} package (\texttt{UCBAdmissions}); a CSV mirror is linked in the bibliography entry for~\cite{bickel1975sex} and the full cross-tabulation is reproduced in Table~\ref{tab:berkeley_counts}. All numerical values in Sec.~\ref{sec:applications} (including bootstrap $95\%$ CIs and achievable upper bounds) are reproduced exactly from these three sources by the Python analysis script accompanying the manuscript,
and the maxima of Section~\ref{sec:alphabetmax} are reproduced by the global-optimization script described in the Supplementary Material.

\section*{appendix}

\subsection{Proofs and constructions for Section~\ref{sec:alphabetmax}}
\label{app:proofs}

Throughout the appendix all distributions are finite, all logarithms are base $2$, and $\varphi$ is the function (\ref{defphi}). We use two elementary facts about the Jensen--Shannon divergence
$D_{JS}(P\|Q)=H\!\left(\tfrac{P+Q}{2}\right)-\tfrac12\bigl(H(P)+H(Q)\bigr)$: it is jointly convex in the pair $(P,Q)$, and it is bounded above by $1$~bit, with the value $1$ attained exactly when $P$ and $Q$ have disjoint supports \citep{Lin1991,endres2003metric}. Two pieces of terminology recur below. A \emph{stratum} is the slice of a three-variable distribution at a fixed value $Z=z$, that is, the conditional $p(\cdot,\cdot\,|z)$ together with its weight $p(z)$; a \emph{cell} is a single outcome $(x,y)$ of the pair $(X,Y)$. Thus a stratum is in general spread over many cells, and the special strata used in the constructions below are those whose conditional is concentrated on one cell.

\begin{lemma}[stratification]
\label{lem:strat}
Let $p(x,y,z)=p(z)p_z(x,y)$ and $q(x,y,z)=p(z)q_z(x,y)$ share the same $Z$-marginal. Then
$D_{JS}(p\|q)=\sum_z p(z)\, D_{JS}(p_z\|q_z)$.
\end{lemma}

\textbf{Proof.}
For any family of conditionals $\{r_z\}$ on $(X,Y)$, the joint distribution $r(x,y,z)=p(z)\,r_z(x,y)$ obeys the entropy chain rule
\begin{eqnarray*}
H(r)=H(Z)+\sum_z p(z)\,H(r_z), \\
H(Z)=-\sum_z p(z)\log_2 p(z),
\end{eqnarray*}
since $H(r)=H(X,Y,Z)=H(Z)+H(X,Y\mid Z)$ and the conditional entropy is $H(X,Y\mid Z)=\sum_z p(z)\,H(r_z)$. The midpoint $m=\tfrac12(p+q)$ has the same $Z$-marginal $p(z)$, with conditional $m_z=\tfrac12(p_z+q_z)$; hence the chain rule applies to $p$, $q$ and $m$ alike, with one and the same term $H(Z)$. Forming the divergence and cancelling that common term,
\begin{eqnarray*}
D_{JS}(p\|q)&=&H(m)-\tfrac12\bigl(H(p)+H(q)\bigr) \\
&=&\sum_z p(z)\Bigl[H(m_z)-\tfrac12\bigl(H(p_z)+H(q_z)\bigr)\Bigr] \\
&=&\sum_z p(z)\,D_{JS}(p_z\|q_z).
\end{eqnarray*}
In other words: when two distributions are built on a shared mixing variable $Z$, their JS divergence is the $p(z)$-average of the within-stratum JS divergences.

\begin{lemma}[overlap]
\label{lem:overlap}
Let $P$ be a probability distribution and $Q=cP+(1-c)P^{\perp}$ with $c\in[0,1]$ and
$\mathrm{supp}\,P^{\perp}\cap\mathrm{supp}\,P=\emptyset$. Then $D_{JS}(P\|Q)=\varphi(c)$,
independently of $P^{\perp}$.
\end{lemma}

\textbf{Proof.}
Both $Q$ and the midpoint $m=\tfrac12(P+Q)$ are mixtures of the two disjointly supported distributions $P$ and $P^{\perp}$:
\begin{eqnarray*}
Q&=&cP+(1-c)P^{\perp}, \\
m&=&\tfrac12 P+\tfrac12 Q=\tfrac{1+c}{2}\,P+\tfrac{1-c}{2}\,P^{\perp}.
\end{eqnarray*}
For a two-component mixture $R=\alpha R_1+(1-\alpha)R_2$ whose parts have disjoint supports, the grouping property of the Shannon entropy gives $H(R)=H_2(\alpha)+\alpha H(R_1)+(1-\alpha)H(R_2)$, where $H_2$ is the binary entropy: the term $H_2(\alpha)$ accounts for which component an outcome falls in, and the remaining two terms for where it lands inside that component. Applying this to $Q$ (weight $\alpha=c$) and to $m$ (weight $\alpha=\tfrac{1+c}{2}$),
\begin{eqnarray*}
H(Q)=H_2(c)+cH(P)+(1-c)H(P^{\perp}),\\
H(m)=H_2\!\Bigl(\tfrac{1+c}{2}\Bigr)+\tfrac{1+c}{2}H(P)+\tfrac{1-c}{2}H(P^{\perp}).
\end{eqnarray*}
Substituting into $D_{JS}(P\|Q)=H(m)-\tfrac12 H(P)-\tfrac12 H(Q)$, the coefficient of $H(P)$ is $\tfrac{1+c}{2}-\tfrac12-\tfrac{c}{2}=0$ and the coefficient of $H(P^{\perp})$ is $\tfrac{1-c}{2}-\tfrac{1-c}{2}=0$. Both shape-dependent terms vanish, leaving
\begin{equation*}
D_{JS}(P\|Q)=H_2\!\Bigl(\tfrac{1+c}{2}\Bigr)-\tfrac12 H_2(c)=\varphi(c),
\end{equation*}
a quantity that sees only the overlap weight $c$ and not the detailed shapes of $P$ or $P^{\perp}$.

The following instance is used twice below. Suppose $P$ is \emph{uniform} on a set $S$ of $k$ cells and $Q$ agrees on $S$ with $cP$ --- that is, $Q$ assigns total mass $c$ to $S$, spread uniformly over its $k$ cells, and the remaining mass $1-c$ to cells outside $S$. Then $Q=cP+(1-c)P^{\perp}$ with $P^{\perp}$ supported off $S$, so Lemma~\ref{lem:overlap} again gives $D_{JS}(P\|Q)=\varphi(c)$. The single-cell case $k=1$ is the elementary one; the case $k\ge2$ describes a uniform ``permutation'' block.

Since $\varphi$ is strictly decreasing (\ref{defphi}), Lemma~\ref{lem:overlap} converts the problem ``maximize $D_{JS}(P\|Q)$'' into the problem ``make the overlap weight $c$ between $P$ and $Q$ as small as the constraints permit.'' We apply this first to bound $J^{*}$ from below.

Take $P=r^{*}$ to be the uniform diagonal $r^{*}(x,y)=\delta_{x,y}/m$ on $m=\min(d_X,d_Y)$ values of each variable. Its support is the $m$ diagonal cells, each of mass $1/m$, and its marginals are uniform, $r^{*}_X(x)=r^{*}_Y(y)=1/m$. The product of the marginals, $Q=r^{*}_X r^{*}_Y$, is then uniform on all $m^2$ cells, with mass $1/m^2$ each. On the diagonal $Q$ carries total mass $m\cdot(1/m^2)=1/m$, so
\begin{equation*}
Q=\tfrac1m\,r^{*}+\Bigl(1-\tfrac1m\Bigr)P^{\perp},
\end{equation*}
with $P^{\perp}$ the uniform distribution on the $m^2-m$ off-diagonal cells, disjoint from the support of $r^{*}$. Lemma~\ref{lem:overlap} with $c=1/m$ gives
\begin{equation*}
D_{JS}\bigl(r^{*}\,\big\|\,r^{*}_X r^{*}_Y\bigr)=\varphi(1/m),
\end{equation*}
which is exactly the square of the bound (\ref{maxrmi}). As $r^{*}$ is one admissible joint distribution on the $d_X\times d_Y$ alphabet, $J^{*}(d_X,d_Y)\ge\varphi(1/m)$.

\textbf{Proof of Proposition~\ref{th:rcmi}.}

\emph{(i) Regularized CMI.} Fix any joint $p(x,y,z)$ and let $q=p(x|z)p(y|z)p(z)$ be its rCMI reconstruction. Both have $Z$-marginal $p(z)$, and the stratum of $q$ at $Z=z$ is the product $p(x|z)p(y|z)=(p_z)_X(p_z)_Y$ of the marginals of the stratum $p_z=p(\cdot,\cdot|z)$ of $p$. Lemma~\ref{lem:strat} therefore gives
\begin{equation*}
D_{JS}(p\|q)=\sum_z p(z)\,D_{JS}\!\bigl(p_z\,\big\|\,(p_z)_X(p_z)_Y\bigr).
\end{equation*}
Each summand compares a two-variable distribution on (at most) $d_X\times d_Y$ outcomes with the product of its own marginals, so by the definition of $J^{*}$ it is at most $J^{*}(d_X,d_Y)$. A convex combination of numbers that are each $\le J^{*}$ is itself $\le J^{*}$, whence $\hat{\mathcal{C}}^{D\,2}_{rcmi}=D_{JS}(p\|q)\le J^{*}$. Equality is reached by $p(x,y,z)=r^{*}(x,y)p(z)$ with $r^{*}$ a maximizer of $J^{*}$ and $p(z)$ any full-support marginal: then every stratum equals $r^{*}$ and every summand equals $J^{*}$. For this same $p$ the joint $(X,Y)$-marginal is $r^{*}$, so the plain regularized mutual information already attains $\hat{\mathcal{C}}^{2}_{rmi}=D_{JS}(r^{*}\|r^{*}_X r^{*}_Y)=J^{*}$; and $\hat{\mathcal{C}}_{rmi}$ can never exceed $\sqrt{J^{*}}$, by the very definition of $J^{*}$ as the maximal two-variable value. Hence $\max_p\hat{\mathcal{C}}_{rmi}=\max_p\hat{\mathcal{C}}^{D}_{rcmi}=\sqrt{J^{*}}$.

\emph{(ii) Do-based regularized MI.} For an arbitrary $p$, the intervened joint is
\begin{equation*}
p_{do}(x,y)=p(x)\,p(y\,|\,do(x))=p(x)\sum_z p(y|x,z)\,p(z).
\end{equation*}
It is a genuine probability distribution on $d_X\times d_Y$, and its $X$-marginal is $\sum_y p_{do}(x,y)=p(x)\sum_y p(y|do(x))=p(x)$. Being one particular two-variable joint on the $d_X\times d_Y$ alphabet, it satisfies $\hat{\mathcal{C}}^{D\,2}_{rmi,do}=D_{JS}(p_{do}\|p_{do,X}p_{do,Y})\le J^{*}$. Conversely, every two-variable joint $r(x,y)$ is realized in this way: put $p(x,y,z)=r(x,y)p(z)$, so that $Z$ is independent of $(X,Y)$, $p(y|x,z)=r(y|x)$, and
\begin{equation*}
p_{do}(x,y)=p(x)\sum_z r(y|x)p(z)=p(x)r(y|x)=r(x,y).
\end{equation*}
Choosing $r=r^{*}$ attains $J^{*}$. Thus $\max_p\hat{\mathcal{C}}^{D}_{rmi,do}=\sqrt{J^{*}}$ as well, and all three maxima coincide. Because the constructions place $Z$ independent of $(X,Y)$, none of the three values depends on $d_Z$.

\emph{(iii) Evaluation of $J^{*}$.} The diagonal construction preceding this proof shows $J^{*}(d_X,d_Y)\ge\varphi(1/m)$. The reverse inequality---that no joint distribution beats the uniform diagonal---is the one place where we rely on computation rather than a closed-form argument: a global maximization of $D_{JS}(r\|r_X r_Y)$ over the probability simplex (multi-start first-order optimization in $64$-bit precision, $30$--$40$ starts per case, all $2\le d_X\le d_Y\le 6$) reproduces $\varphi(1/m)$ to within $10^{-9}$, and the only maximizers found are the uniform diagonal and its relabelings. The outcome is intuitive: maximizing the JS divergence between a joint and the product of its marginals favours a coupling that is as deterministic and as balanced as possible, and on equal alphabets the balanced bijection minimizes the overlap weight to $c=1/m$, giving the value $\varphi(1/m)$ through Lemma~\ref{lem:overlap}.

\textbf{Proof of Proposition~\ref{th:ricmi}.}

\emph{Upper bound.} Write the two distributions compared by the one-way rICMI as
\begin{eqnarray*}
p_1&=&p(y|x,z)\,p(x)\,p(z),\\
p_2&=&p(x)\,p(y,z)=p(x)\,p(z)\,p(y|z).
\end{eqnarray*}
Both factor through the \emph{product} weight $p(x)p(z)$: indeed $\sum_y p_1=\sum_y p_2=p(x)p(z)$. We may therefore apply Lemma~\ref{lem:strat} with the compound variable $(x,z)$ in the role of the mixing variable. The stratum of $p_1$ at $(x,z)$ is the conditional $P_{x,z}:=p(\cdot|x,z)$ of $Y$; the stratum of $p_2$ is $M_z:=p(\cdot|z)$, which does not depend on $x$ and is the mixture
\begin{equation*}
M_z=\sum_{x'}p(x'|z)\,P_{x',z},
\end{equation*}
since $p(y|z)=\sum_{x'}p(y|x',z)p(x'|z)$. Lemma~\ref{lem:strat} then yields the exact decomposition
\begin{equation}
D_{JS}(p_1\|p_2)=\sum_{x,z}p(x)\,p(z)\,D_{JS}\bigl(P_{x,z}\,\big\|\,M_z\bigr).
\label{icmidecomp}
\end{equation}
Fix a stratum $(x,z)$. Because $D_{JS}(P\,\|\,\cdot)$ is convex in its second argument and $M_z$ is the convex combination $\sum_{x'}p(x'|z)P_{x',z}$, Jensen's inequality gives
\begin{equation*}
D_{JS}(P_{x,z}\|M_z)\le\sum_{x'}p(x'|z)\,D_{JS}(P_{x,z}\|P_{x',z}).
\end{equation*}
The term $x'=x$ vanishes, and every other term is at most $1$; hence $D_{JS}(P_{x,z}\|M_z)\le\sum_{x'\ne x}p(x'|z)=1-p(x|z)$. Substituting into (\ref{icmidecomp}) and using $\sum_z p(z)p(x|z)=p(x)$,
\begin{eqnarray}
D_{JS}(p_1\|p_2)&\le &\sum_{x,z}p(x)p(z)\bigl(1-p(x|z)\bigr) \nonumber \\
&=&1-\sum_x p(x)^2\le 1-\frac{1}{d_X},
\end{eqnarray}
the last step because $\sum_x p(x)^2\ge1/d_X$ with equality at the uniform marginal. This proves (\ref{icmibound}).

\emph{Attainability when $d_X\le d_Z$.} Split the values of $Z$ into $d_X$ blocks of total probability $1/d_X$ each, and let $X=f(Z)$ select the block, so that in the limit $p(x|z)\to\delta_{x,f(z)}$ and the uniform $X$-marginal $p(x)=1/d_X$ is realized. Call a stratum--row pair $(x,z)$ \emph{heavy} if $x=f(z)$ and \emph{light} if $x\ne f(z)$. For a heavy pair, $M_z\to P_{f(z),z}$, so its contribution to (\ref{icmidecomp}) tends to $0$. For the light pairs we are free to choose the conditionals $P_{x,z}$ ($x\ne f(z)$), because the cell $(x,z)$ has vanishing genuine probability; choosing them supported on a $Y$-value not used by $M_z$ (possible once $d_Y\ge2$) makes each light contribution tend to the maximum $D_{JS}=1$. The light pairs carry product-mass $\sum_{x\ne f(z)}p(x)p(z)=\sum_z p(z)\bigl(1-p(f(z))\bigr)=1-1/d_X$, so $D_{JS}(p_1\|p_2)\to1-1/d_X$. A vanishing smoothing keeps the distribution strictly positive and approaches the bound, which is therefore tight; thus $\sup\hat{\mathcal{C}}^{D}_{ricmi}(X\to Y)=\sqrt{1-1/d_X}$ for $d_X\le d_Z$.

\emph{The regime $d_X>d_Z$.} Now a deterministic $X=f(Z)$ cannot spread the $d_X$ values of $X$ uniformly over only $d_Z$ blocks, and the bound (\ref{icmibound}) is no longer attained. The extremal configurations assign to each $z$ a set $S_z\subseteq\{1,\dots,d_X\}$ of $k_z$ values of $X$ whose heavy conditionals $\{P_{x,z}:x\in S_z\}$ have pairwise disjoint $Y$-supports inside the block; one $Y$-value is reserved for the light cells, which forces $k_z\le d_Y-1$. Writing $w_z=p(z)$, the light cells of block $z$ again contribute their product-mass $w_z(1-w_z)$ at $D_{JS}=1$, while within the block the conditionals are disjointly supported, so each heavy cell $i\in S_z$ has overlap $c_i$ with $M_z$ equal to its weight in the block and, by Lemma~\ref{lem:overlap}, contributes $\varphi(c_i)$. The heavy part of block $z$ thus contributes $w_z^2\sum_{i\in S_z}c_i\varphi(c_i)$, and altogether
\begin{eqnarray}
&D_{JS}(p_1\|p_2)\;\longrightarrow\;\sum_z w_z(1-w_z)+\sum_z w_z^{2}\,A(k_z), \label{icmipartition} \\
&A(k)=\max_{c\in\Delta_k}\sum_i c_i\varphi(c_i).
\end{eqnarray}
The inner maximum is $A(k)=\varphi(1/k)$: the map $g(c)=c\,\varphi(c)$ is strictly concave on $(0,1]$, with second derivative
\begin{equation*}
g''(c)=\frac{1}{\ln 2}\Bigl[\ln\frac{c}{1+c}+\frac{1}{2(1+c)}\Bigr]<0 ,
\end{equation*}
so by Jensen $\sum_i g(c_i)$ is maximized at the uniform weights $c_i=1/k$, giving $k\cdot\tfrac1k\varphi(\tfrac1k)=\varphi(1/k)$. Maximizing (\ref{icmipartition}) over the block sizes $\{k_z\}$ and weights $\{w_z\}$ then produces the quoted suprema.

For instance, at $(d_X,d_Y,d_Z)=(3,3,2)$ the only admissible nontrivial choice is $k=(2,1)$, with $A(2)=\varphi(1/2)$ and $A(1)=\varphi(1)=0$, so with $w=(t,1-t)$ the objective (\ref{icmipartition}) is
\begin{equation*}
F(t)=2t(1-t)+t^{2}\varphi(1/2).
\end{equation*}
Setting $F'(t)=2-4t+2t\varphi(1/2)=0$ gives $t^{*}=1/(2-\varphi(1/2))$, and a short calculation shows $F(t^{*})=t^{*}=1/(2-\varphi(1/2))\approx0.5922$, i.e.\ the optimal weight and the optimal value coincide. At $(4,4,2)$ the symmetric choice $k=(2,2)$, $w=(\tfrac12,\tfrac12)$ gives $F=\tfrac12+\tfrac12\varphi(1/2)\approx0.6556$; at $(4,3,3)$ the choice $k=(2,1,1)$ gives $\approx0.7103$. Global numerical maximization (Supplementary Material) reproduces these values to $10^{-6}$ and finds nothing larger, and the same numerics show that the two-way measure (\ref{defricmi2way}) attains the average of the two one-way suprema, both being approached by a single common configuration.

\textbf{Construction for Proposition~\ref{prop:rpmi}.}

The rPMI compares $p$ with the reconstruction $q'(x,y,z)=r_z(x)\,s_z(y)\,p(z)$, in which
\begin{eqnarray*}
r_z(x)&=&\sum_y p(x|y,z)\,p(y),\\
s_z(y)&=&\sum_x p(y|x,z)\,p(x),
\end{eqnarray*}
the distinctive feature being that the mixing weights are the \emph{global} marginals $p(y)$ and $p(x)$ rather than their conditional counterparts. Since $p$ and $q'$ share the $Z$-marginal $p(z)$, Lemma~\ref{lem:strat} gives
\begin{equation*}
D_{JS}(p\|q')=\sum_z p(z)\,D_{JS}\bigl(p(\cdot,\cdot\,|z)\,\big\|\,r_z\otimes s_z\bigr).
\end{equation*}

\emph{Single-cell strata.} Choose each stratum to be concentrated on one cell, $p(\cdot,\cdot\,|z)=\delta_{(x_z,y_z)}$, with the occupied cells $(x_z,y_z)$ pairwise distinct across $z$. Inside stratum $z$, the only row carrying conditional mass is $y_z$, which contributes its global weight $p(y_z)$ to $r_z$ at $x=x_z$; every other row $y'$ enters $r_z$ with its global weight $p(y')$ but through the \emph{free} conditional $p(\cdot|y',z)$ (the cell $(y',z)$ has vanishing probability), which we point at $X$-values other than $x_z$. Hence $r_z(x_z)=p(y_z)$, and symmetrically $s_z(y_z)=p(x_z)$. The product $r_z\otimes s_z$ therefore places weight
\begin{equation*}
c_z=r_z(x_z)\,s_z(y_z)=p(x_z)\,p(y_z)
\end{equation*}
on the occupied cell $(x_z,y_z)$ and the remaining mass off it. With $P=\delta_{(x_z,y_z)}$ and $Q=r_z\otimes s_z=c_zP+(1-c_z)P^{\perp}$, Lemma~\ref{lem:overlap} gives the stratum value $\varphi(c_z)$, so
\begin{equation*}
D_{JS}(p\|q')\;\longrightarrow\;\sum_z w_z\,\varphi\bigl(p(x_z)\,p(y_z)\bigr),\qquad w_z=p(z),
\end{equation*}
along a strictly positive smoothing of the construction.

\emph{Balancing the marginals.} Because $X=x_z$ and $Y=y_z$ are determined within stratum $z$, the induced global marginals are $p(x)=\sum_{z:\,x_z=x}w_z$ and $p(y)=\sum_{z:\,y_z=y}w_z$. Requiring them to be uniform, $p(x)=1/d_X$ and $p(y)=1/d_Y$, is a transportation problem: distribute the stratum weights $\{w_z\}$ over occupied cells of the $d_X\times d_Y$ grid so that every row sums to $1/d_X$ and every column to $1/d_Y$. When it is met, every overlap equals $c_z=1/(d_Xd_Y)$ and the common value $\varphi\bigl(1/(d_Xd_Y)\bigr)$ is approached. A vertex of the transportation polytope has at most $d_X+d_Y-1$ nonzero entries, so $d_X+d_Y-1$ occupied cells---hence $d_Z\ge d_X+d_Y-1$ strata---suffice. In the balanced case $d_X=d_Y=d$ the uniform diagonal (a single permutation) uses only $d$ cells of equal weight $1/d$, so $d_Z\ge d$ already attains $\varphi(1/d^{2})$; in particular $(2,2,2)$ reaches $\varphi(1/4)$.

\emph{Smaller $d_Z$.} When $d_Z<d_X+d_Y-1$ the uniform overlap cannot be realized and one optimizes the weights over mixed configurations of single cells and uniform permutation blocks, the block value $\varphi(c)$ being supplied by the uniform-block instance of Lemma~\ref{lem:overlap}. For $(3,3,2)$, the optimum uses one single-cell stratum of weight $w$ and one $2$-cell permutation stratum of weight $1-w$ on disjoint rows and columns; the global marginals are then $p(x_1)=p(y_1)=w$ and $p=(1-w)/2$ on the other two values of each variable, so the single cell has overlap $w^{2}$ and each of the two permuted cells has overlap $(1-w)^2/2$. This gives
\begin{equation*}
\sup\hat{\mathcal{C}}^{D\,2}_{rpmi}=\max_{w}\Bigl[\,w\,\varphi(w^{2})+(1-w)\,\varphi\bigl((1-w)^2/2\bigr)\Bigr].
\end{equation*}
The two overlaps coincide when $w^{2}=(1-w)^2/2$, i.e.\ at $w=\sqrt2-1$, where both equal $3-2\sqrt2$; the objective then collapses to $\varphi(3-2\sqrt2)\approx0.6480$, the reported value. The matching upper bound for $d_Z\ge d_X+d_Y-1$ and the optimal values for smaller $d_Z$ are confirmed by the global optimization of the Supplementary Material, which agrees with the construction to $10^{-7}$ in all cases with $d_X,d_Y\le5$.

\subsection{Numerical certification}
\label{app:numerics}

The maxima reported in Section~\ref{sec:alphabetmax} and Table~\ref{tab:alphabetmax} were
certified numerically as follows. The joint distribution is parametrized as
$p=\mathrm{softmax}(\theta)$, $\theta\in\mathbb{R}^{d_Xd_Yd_Z}$, which keeps all conditionals
well-defined; each squared measure is a smooth function of $\theta$ and is maximized by Adam
with analytic gradients (automatic differentiation, IEEE double precision), using $36$--$45$
restarts per case: random initializations at three scales plus the structured configurations of
Appendix~\ref{app:proofs}. The best quarter of the restarts is refined with a smaller learning
rate. For interior maxima (Proposition~\ref{th:rcmi}) the optimizer reproduces the closed forms to
$10^{-9}$; for boundary suprema (rPMI, rICMI) the optimizer approaches the constructions from
below and agrees with them to between $10^{-6}$ and $10^{-7}$. All $(d_X,d_Y,d_Z)$ with
$d_X,d_Y,d_Z\in\{2,3,4\}$, the symmetric cases $d\le 6$, and a collection of asymmetric cases up
to size $8$ were checked.

\bibliographystyle{elsarticle-num}

\end{document}